\PassOptionsToPackage{dvipsnames}{xcolor}

\documentclass[sigconf, nonacm]{acmart}

\usepackage{multicol}

\newcommand\vldbdoi{XX.XX/XXX.XX}
\newcommand\vldbpages{XXX-XXX}
\newcommand\vldbvolume{14}
\newcommand\vldbissue{1}
\newcommand\vldbyear{2020}
\newcommand\vldbauthors{\authors}
\newcommand\vldbtitle{\shorttitle} 
\newcommand\vldbavailabilityurl{https://github.com/yannramusat/TPG}
\newcommand\vldbpagestyle{plain} 

\usepackage[frozencache=true,cachedir=.]{minted}
\AtBeginEnvironment{minted}{%
  \renewcommand{\fcolorbox}[4][]{#4}}

\usepackage{etoolbox}

\usepackage{xcolor}
\definecolor{CypherGreen}{RGB}{0, 120, 0}
\usepackage{listings}
\usepackage{cprotect}

\usepackage{enumitem}

\makeatletter
\AtBeginEnvironment{noerr}{\dontdofcolorbox}
\def\dontdofcolorbox{\renewcommand\fcolorbox[4][]{##4}}
\makeatother

\usepackage{tikz}
\usetikzlibrary{calc}

\newcommand{\arr}[3]{%
  \begin{tikzpicture}[baseline={(0,-0.1)},yscale=0.7]%
  \node (up) at (0,.275) {\( \scriptstyle #2 \) };
  \node (down) at (0,-.25) { \(\scriptstyle #3 \) };
  \path let \p1=(up.west),\p2=(down.west) in coordinate (left) at ({min(\x1,\x2)},0);
  \path let \p1=(up.east),\p2=(down.east) in coordinate (right) at ({max(\x1,\x2)},0);
  \draw[#1] (left) -- (right);
\end{tikzpicture}}

\usepackage[bibliography=common, appendix=append]{apxproof}
\newtheoremrep{theorem}{Theorem}[section]
\newtheoremrep{lemma}[theorem]{Lemma}
\newtheoremrep{proposition}[theorem]{Proposition}

\usepackage{algorithm}
\usepackage[noend]{algpseudocode} 

\usepackage{caption}
\usepackage{subcaption}

\begin{document}
\title{Transforming Property Graphs (Extended Version)\footnote{A shorter version of this paper has been accepted for publication in VLDB 2024.}}

\author{Angela Bonifati}
\affiliation{%
  \institution{Lyon 1 Univ., Liris CNRS \& IUF}
}
\email{angela.bonifati@univ-lyon1.fr}

\author{Filip Murlak}
\affiliation{%
  \institution{Univ. of Warsaw}
}
\email{fmurlak@mimuw.edu.pl}

\author{Yann Ramusat}
\orcid{0000-0001-5109-3700}
\affiliation{%
  \institution{Lyon 1 Univ., Liris CNRS}
}
\email{yann.ramusat@liris.cnrs.fr}

\begin{abstract}
    In this paper, we study a declarative framework for specifying transformations of property graphs. 
    In order to express such transformations, we leverage queries formulated in the Graph Pattern Calculus (GPC), which is an abstraction of the common core of recent standard graph query languages, GQL and SQL/PGQ.
    In contrast to previous frameworks targeting graph topology only, we focus on the impact of data values on the transformations---which is crucial in addressing users' needs. 
    In particular, we study the complexity of checking if the transformation rules do not specify conflicting values for properties, and we show this is closely related to the  satisfiability problem for GPC. 
    We prove that both problems are \textsc{PSpace}-complete.

    We have implemented our framework in openCypher. We show the flexibility and usability of our framework by leveraging 
    an existing data integration benchmark, adapting it to our needs.
    We also evaluate the incurred overhead of detecting potential inconsistencies at run-time, and the impact of several optimization tools in a Cypher-based graph database, 
    by providing a comprehensive comparison of different implementation variants.
    The results of our experimental study show that our framework exhibits large practical benefits for transforming property graphs compared to ad-hoc transformation scripts.
\end{abstract}

\maketitle

\pagestyle{\vldbpagestyle}
\begingroup\small\noindent\raggedright\textbf{PVLDB Reference Format:}
\vldbauthors. \vldbtitle. PVLDB, \vldbvolume(\vldbissue): \vldbpages, \vldbyear.\\
\href{https://doi.org/\vldbdoi}{doi:\vldbdoi}
\endgroup
\begingroup
\renewcommand\thefootnote{}\footnote{
\noindent
}\addtocounter{footnote}{-1}\endgroup

\ifdefempty{\vldbavailabilityurl}{}{
\vspace{.3cm}
\begingroup\small\noindent\raggedright\textbf{PVLDB Artifact Availability:}\\
The source code, data, and/or other artifacts have been made available at \url{\vldbavailabilityurl}.
\endgroup
}

\section{Introduction}
\label{introduction}

Query languages for property graphs---those supported by existing systems, such as Neo4j’s openCypher~\cite{francis_cypher_2018} or Oracle’s PGQL~\cite{10.1145/2960414.2960421}, 
and those codified in international standards, such as GQL and SQL/PGQ~\cite{francis_researchers_2023}---define their semantics in terms of sets of tuples.
This is inadequate for data interoperability tasks such as data migration or data integration, where outputs of some queries are to be fed directly to other queries. To support this kind of \emph{composability}, queries should be able to output property graphs rather than sets of tuples. Such queries can be seen as \emph{transformations}, turning an input property graph into an output property graph. 

Interoperability of graph data has received little attention so far, compared to the relational and XML data models~\cite{10.5555/1941440}. 
Notable research in the area \cite{10.1145/2448496.2448520, 10.1145/3584372.3588654} relies on the simplified graph data model that had been devised to provide the foundations for querying the topology of graphs with formalisms such as conjunctive regular path queries (CRPQs)~\cite{10.1145/2463664.2465216} or regular queries~\cite{vardi_theory_2016}. As the simplified graph data model ignores the presence of properties (key-value pairs stored in nodes and edges), it is too far from the property graph models used in graph databases such as Neo4j or Tigergraph, and cannot be a foundation for practical solutions. These currently 
rely on opaque external libraries, such as Neo4j's APOC~\cite{apoc}, or involve complex handcrafted queries~\cite{graphacademy}, as illustrated below.  

\begin{figure}[!ht]
    \centering
    \begin{subfigure}[b]{0.525\textwidth}
        \captionsetup{justification=centering}
        \centering
        \includegraphics[width=7cm]{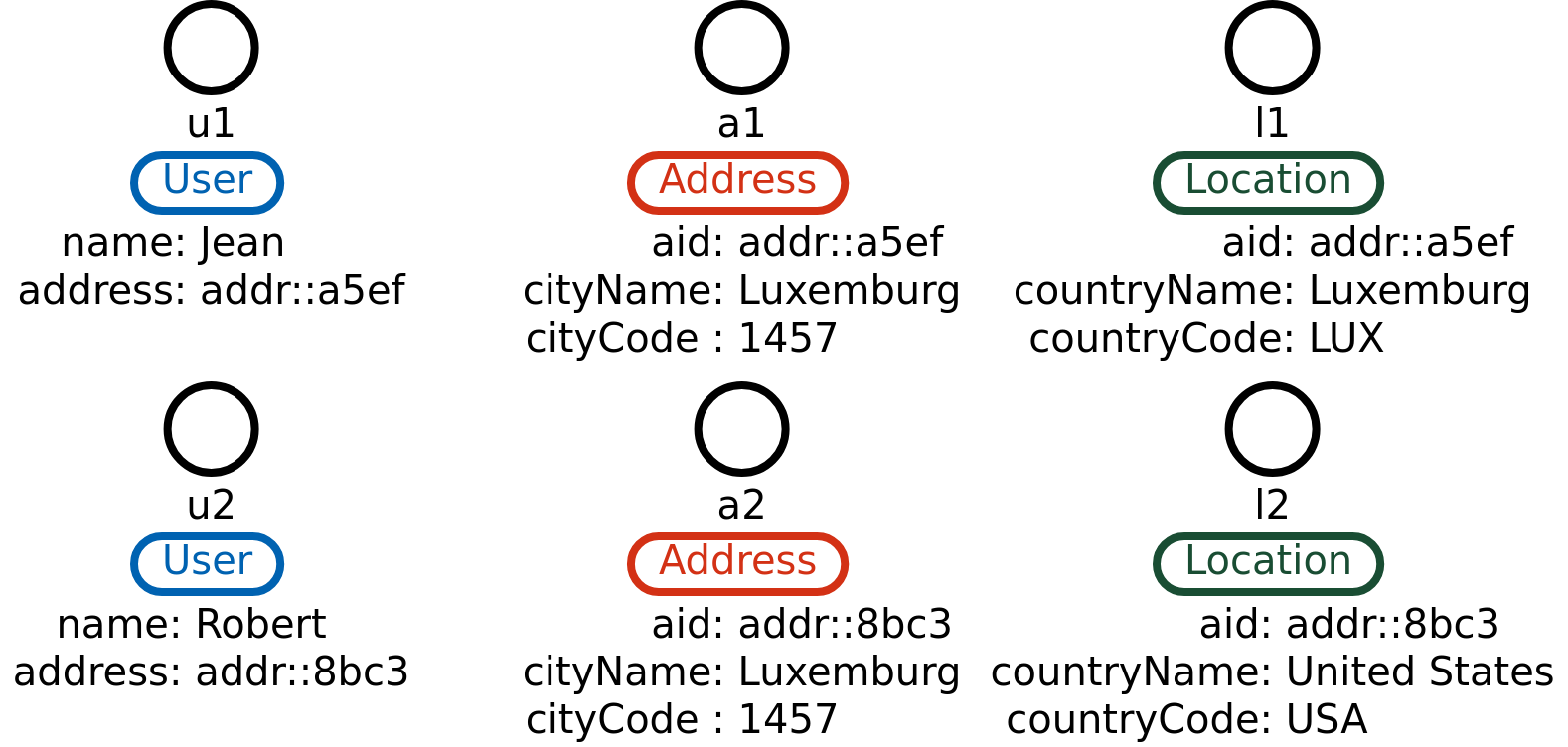}
        \caption{Input property graph $G$ containing ingested relational data.}
        \label{re:input}
    \end{subfigure}
    \begin{subfigure}[b]{0.525\textwidth}
        \bigskip
        \raggedleft
        \captionsetup{justification=centering}
        \begin{minted}[xleftmargin=2em, linenos=true, fontsize=\footnotesize, escapeinside=!!]{cypher}
MATCH (u:User)
MATCH (a:Address) WHERE a.aid = u.address
MATCH (l:Location) WHERE l.aid = u.address
WITH u, collect(a) AS Addresses, collect(l) AS Locations
CREATE (p:Person) !\label{naive:person}!
SET p.name = u.name
WITH p, Addresses, Locations
UNWIND Addresses AS a
MERGE (ci:City {name: a.cityName})             !\label{naive:city}!
SET ci.code = a.cityCode
MERGE (p)-[:HasAddress]->(ci)
WITH p, Locations
UNWIND Locations AS l
MERGE (co:Country {name: l.countryName}) !\label{naive:country}!
SET co.code = l.countryCode
MERGE (p)-[:HasLocation]->(co)
        \end{minted}
        \caption{Ad-hoc transformation script in openCypher.}
        \label{re:cypher}
    \end{subfigure}%

    \begin{subfigure}[b]{0.525\textwidth}
        \bigskip
        \captionsetup{justification=centering}
        \centering
        \includegraphics[width=6.3cm]{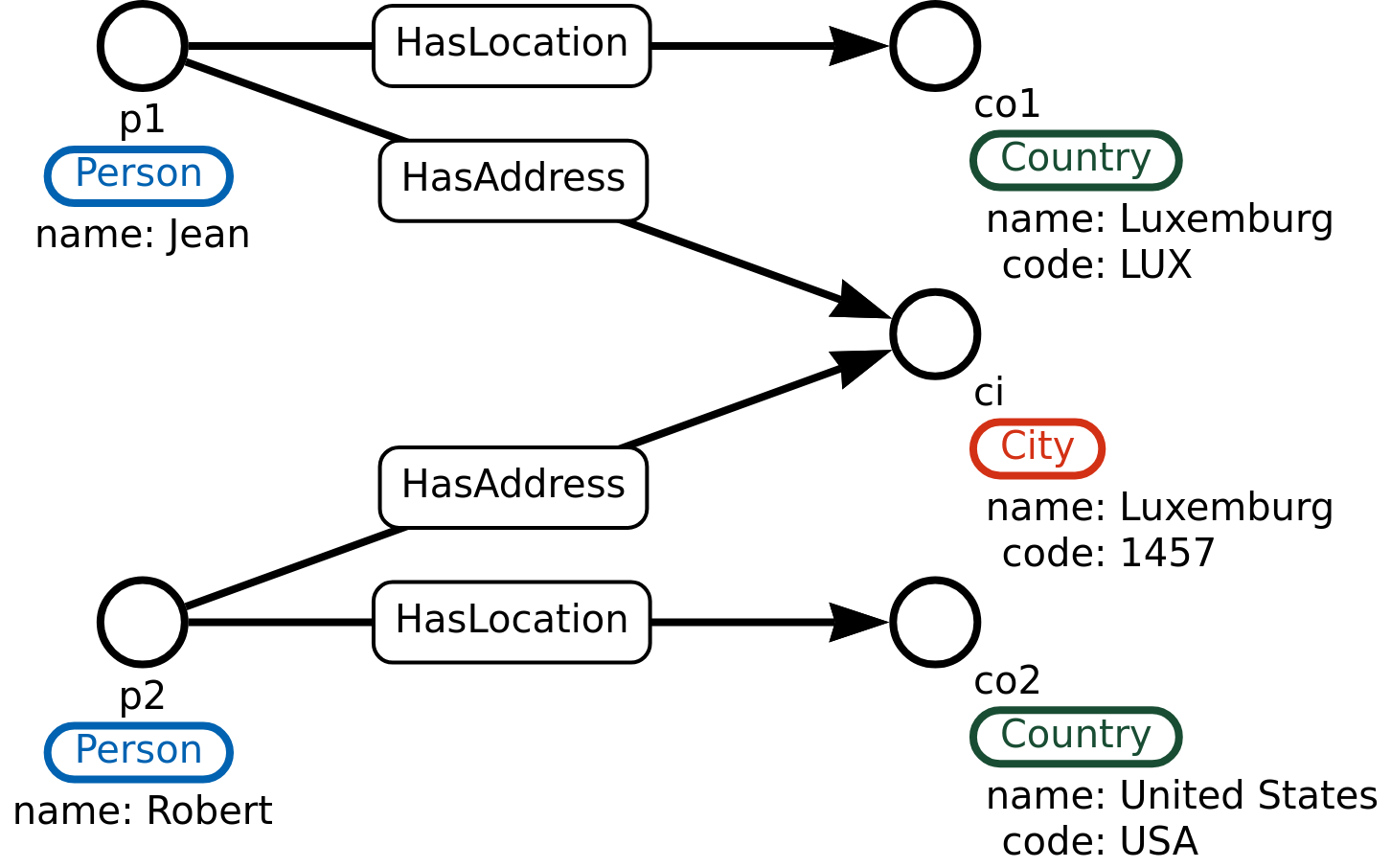}
        \caption{Resulting output property graph $H$.} 
        \label{re:output}
    \end{subfigure}
    \caption{Ad-hoc transformation of raw ingested data.}
    \vspace{-0.5em}
    \label{fig:motivating-scenario}
\end{figure}

\begin{example}
\label{ex:motivating-example}
Figure~\ref{fig:motivating-scenario} illustrates a graph transformation scenario, in which a user has imported relational data into the popular Neo4j graph database and would like to reshape it into a semantically meaningful property graph instance, to  facilitate navigational querying.
The relational data consists of three tables, 
\begin{gather*}
\mathsf{User}(\underline{name}, \underline{address})\,, \; 
\mathsf{Address}(\underline{aid}, cityName,cityCode)\,,\\\mathsf{Location}(\underline{aid}, countryName,countryCode)\,,
\end{gather*}
with primary keys consisting of the underlined attributes, and having two foreign keys: $aid$ references $address$ in $\mathsf{User}$ from both $\mathsf{Address}$ and $\mathsf{Location}$.

Figure~\ref{fig:motivating-scenario} (\subref{re:input}) shows a rudimentary property graph obtained after importing the relational data, using a generic ingestion method, such as  Cypher's \mintinline{cypher}|LOAD CSV| clause.
In the resulting property graph, each node represents a single tuple of the relational instance, with the relation's name represented as the label, and the attributes stored in the node's properties.
Note that there are no edges in this property graph: relationships between places, locations, and users are represented by way of foreign keys, just like in the original relational instance. Needless to say, this is not the best way to represent a relational instance as a property graph. 

\looseness=-1
The user now wants to transform the instance in Figure~\ref{fig:motivating-scenario}~(\subref{re:input}) into one that makes better use of the property graph data model by facilitating navigational operations in queries like ``Which people live in the same city as Jean?''. The user intends to create a node for each person, city, and country, and replace foreign key references with explicit relationships. Figure~\ref{fig:motivating-scenario} (\subref{re:cypher}) shows an implementation of this transformation in openCypher that closely follows a graph refactoring solution described in Neo4j's GraphAcademy~\cite{graphacademy}.
The reader will notice how difficult it is to relate the constructs of this query to the informal specification above. 
Even just making sense of the  \mintinline{cypher}|MERGE| clauses interleaved with implicit grouping and list manipulations (\mintinline{cypher}|UNWIND| and \mintinline{cypher}|collect|) is a daunting task for an unacquainted user. 
But the query leverages other advanced idioms too.
For instance, in Line~\ref{naive:person}, the script creates as many nodes of type \mintinline{cypher}|Person| as there are rows output by the previous \mintinline{cypher}|WITH| clause: one for each $u$, due to implicit grouping.
In line~\ref{naive:city}, the script generates one \mintinline{cypher}|City| node for each \emph{distinct} value found in property \mintinline{cypher}|cityName| across all $a$'s; this is because the property \mintinline{cypher}|name| is specified as \mintinline{cypher}|a.cityName| in the \mintinline{cypher}|MERGE| clause. Similarly, in line~\ref{naive:country}, a single \mintinline{cypher}|Country| node is created for each \emph{distinct} value found in property \mintinline{cypher}|countryName|. 

Figure~\ref{fig:motivating-scenario}~(\subref{re:output}) shows the output property graph obtained by running the script on the input property graph from  Figure~\ref{fig:motivating-scenario}~(\subref{re:input}). 
It reveals that the ad-hoc transformation fails to account for the fact that cities are weak entities that cannot be identified by their name alone, and conflates Luxemburg in Europe with Luxemburg in the US. Detecting such errors is hard because openCypher lacks a transparent mechanism for specifying identities of created elements.
\hfill $\blacktriangleleft$
\end{example}

\looseness=-1
As we have seen, ad-hoc transformation scripts are error-prone and hard to interpret and analyze. 
Moving away from handcrafted implementations to declarative specifications has long been recognized as pivotal for solving data programmability problems~\cite{bernstein_model_2007}. 
The aim of this work is to lay the theoretical foundations for the declarative specification of property graph transformations,
and facilitate practical solutions for turning such specifications into executable scripts in modern property graph query languages.
Constraint-based, fully declarative formalisms, such as  schema mappings for relational~\cite{fagin_data_2005,10.1145/1061318.1061323, bellahsene2011schema} and graph~\cite{10.1145/2448496.2448520,10.1145/3034786.3056113} data, allow multiple target solutions, leading to ambiguous transformations ~\cite{10.5555/1182635.1164136, DBLP:conf/sigmod/BonifatiCCT17}.
For property graphs, this makes the schema mapping problem  undecidable even under strong restrictions \cite{10.1145/3034786.3056113}. We avoid this problem by focusing on transformations that return a unique, well-defined output instance for each input instance, thus facilitating direct execution.

\looseness=-1
We propose a rule-based formalism that allows the user to describe the output property graph based on the input property graph, by specifying not only labels, properties, and relationships between output elements, but also their identities.
The formalism builds upon the \emph{Graph Pattern Calculus} (GPC)~\cite{10.1145/3584372.3588662}, which is an abstraction of the common graph pattern matching features of  GQL and SQL/PGQ~\cite{DBLP:conf/sigmod/DeutschFGHLLLMM22}. GPC is adequate in terms of expressive power: it has ample facilities for querying properties and even on property-less graphs it goes well beyond classical formalisms such as RPQs and CRPQs. It is suitable for theoretical investigation owing to its concise syntax and rigorous semantics. It should also keep our proposal future-proof by ensuring out-of-the-box compatibility with the expected implementations of these standards.
Until then, we can rely on the already implemented graph query languages, such as Neo4j’s openCypher~\cite{francis_cypher_2018} or Oracle’s PGQL~\cite{10.1145/2960414.2960421}, which were a strong inspiration for GQL.
Indeed, the actual query language used in the rules can be seen as a parameter of the framework.

In contrast to the purely topological formalism of~\cite{10.1145/3584372.3588654}, specifications of property-aware transformations can easily become \emph{inconsistent}, when they attempt to specify two different values for the same property of a given element. Detecting such conflicts naturally comes to the foreground of static analysis. As we show, this problem is tightly connected to the \emph{satisfiability} problem for GPC+ (GPC extended with union and projection, also introduced in~\cite{10.1145/3584372.3588662}), which is to  decide if there is a property graph satisfying a given GPC+ query. Exploiting this connection, we establish tight complexity bounds for both these problems, showing that they are  \textsc{PSpace}-complete. To the best of our knowledge, this is the first static-analysis result on GPC. Given that query satisfiability is the work horse of static analysis throughout database theory, we believe that with the adoption of the GQL standard our result will find other uses. An immediate consequence for property graph transformations is that consistency cannot be checked statically due to the prohibitive cost, and conflicts must be handled dynamically, during the execution of the transformation.  

\looseness=-1 
In order to prove that our formalism can serve as a foundation for practical data interoperability solutions, we provide a proof-of-concept implementation. 
As no existing query engine supports GQL yet, we rely on the Neo4j's open-source implementation of \mbox{openCypher}~\cite{francis_cypher_2018, green_updating_2019}, which offers most of the functionalities of GQL described in~\cite{francis_researchers_2023}. 
We study the case when the rules are provided by the users and describe a generic, easily automated method of translating these rules into executable \mbox{openCypher} scripts, and apply it manually to selected realistic property graph transformations derived from real-world data integration scenarios of the iBench benchmark suite~\cite{arocena_ibench_2015}.
We perform a comprehensive experimental study gauging the efficiency of conflict detection and the effect of rule order and various optimizations on several implementation variants. We confirm that our implementation performs well in all scenarios and scales to large input data.
We also demonstrate that our framework can be successfully applied in a concrete data integration scenario on real-world data~\cite{ICIJ-github}, and report the results of a small-scale user study confirming that our framework enhances readability and usability of transformations.

In summary, our main contributions are the following. 
\begin{itemize}
    \item We propose a comprehensive declarative formalism for specifying transformations of property graphs, compatible with SQL/PGQ and the upcoming GQL standard.
    \item We identify consistency as a key static-analysis problem, and show that it is  interreducible
    with satisfiability of GPC+ queries and that both are \textsc{PSpace}-complete.
    \item We provide a proof-of-concept implementation of our formalism in openCypher, and apply it to realistic scenarios of graph-shaped data transformations.
    \item We show experimentally that our solution scales to large input data, handles on-the-fly conflict detection with low overhead, and enhances readability and usability, without sacrificing preformance.
\end{itemize}

The rest of paper is organized as follows. In Section~\ref{sec:preliminaries}, we recall the property graph data model along with GPC. In Section~\ref{sec:pgt}, we 
give syntax and semantics of our graph transformation formalism.
In Section~\ref{sec:conflicts}, we discuss the consistency in relation with satisfiability of GPC+ queries, and establish the complexity bounds. 
In Section~\ref{sec:translation},  we describe our-proof-of-concept implementation. 
In Section~\ref{sec:experiments}, we present both the experiments and the user study.
In Section~\ref{sec:rw} and in Section~\ref{sec:conclusion}, we discuss the related work and conclude the paper.

\section{Preliminaries}
\label{sec:preliminaries}

We briefly introduce the basic concepts of the property graph data model and the \emph{Graph Pattern Calculus} (GPC) that we use in this paper. 
We mostly follow the definitions from~\cite{10.1145/3584372.3588662}.

\begin{toappendix}
    \section{Preliminaries}
    \label{apx:preliminaries}
\end{toappendix}

\subsection{Data model}
\label{subsec:data-model}

\begin{toappendix}
    \subsection{Data model}
    \label{apx:data-model}
    We introduce thereafter further basic notation around the data model of a property graph. 
    
    A \emph{path} $p \coloneqq (n_0, e_1, n_1, \dots, e_k, n_k)$ is an alternating sequence of node and edges, which starts and ends with nodes. 
    Given a path $p$ we denote by $\mathsf{src}(p)$ and $\mathsf{tgt}(p)$ the first and last node of $p$; in this case, $\mathsf{src}(p) = n_0$ and $\mathsf{tgt}(p) = n_k$.
    $\mathsf{len}(p) \in \mathbb{N}$ is the \emph{length} of $p$, i.e, the number of edges in the path; and if $\mathsf{len}(p) = 0$ then the path consists of a single node which is both the source and the target.
    We can define as usual the \emph{concatenation} $p \cdot p'$ of two paths $p$ and $p'$ whenever $\mathsf{tgt}(p) = \mathsf{src}(p')$.
\end{toappendix}

Conforming to the formal specification originating from~\cite{10.1145/3584372.3588662}, a \emph{property graph} $G$ is a tuple
$\langle N, E, \lambda, \mathsf{src}, \mathsf{tgt}, \delta \rangle$
where $\mathcal{O}$, $\mathcal{L}$, $\mathcal{K}$ and $\mathsf{Const}$ are disjoint countable sets of object identifiers (ids), labels, keys (also called property names) 
and constants (data values), and
\begin{itemize}
    \item $N \subset \mathcal{O}$ is the finite set of node ids in $G$;
    \item $E \subset \mathcal{O}$ is the finite set of edge ids;
    \item $N$ and $E$ are disjoint;
    \item $\lambda : N \cup E \to 2^\mathcal{L}$ is a labeling function that associates to every id a (possibly empty) finite set of labels;
    \item $\mathsf{src}, \mathsf{tgt} : E \to N$ define the source and target of each edge;
    \item $\delta : (N \cup E) \times \mathcal{K} \to \mathsf{Const}$ is a finite-domain partial function that associates a constant with an id and a key from $\mathcal{K}$.
\end{itemize}
The node ids and edge ids will be respectively called the \emph{nodes} and \emph{edges} of the property graph.

That is to say, a property graph is a \emph{multigraph} in the sense that two vertices may be connected by more than one edge, even with these edges having the same label(s), and that loops are permitted. 
All the \emph{elements} of the database (the nodes and the edges) store a finite set of property-value pairs, represented by $\delta$.

A property graph is presented in Figure~\ref{fig:motivating-scenario}~(\subref{re:output}).
It contains information about peoples' location such as the $\mathsf{City}$ and the $\mathsf{Country}$ they live in.
We see that it contains one node with label $\mathsf{City}$ and two nodes with label $\mathsf{Country}$; two edges with label $\mathsf{HasLocation}$ and two edges with label $\mathsf{HasAddress}$;
all nodes have property $name$; and all nodes with label $\mathsf{City}$ or $\mathsf{Country}$ have an additional property $code$.
Annotations $\mathsf{p1}, \mathsf{p2}, \dots, \mathsf{co2}$ are node identifiers;
edge identifiers are not shown.

\subsection{Graph Pattern Calculus}
\label{subsec:prelim-gpc}

\begin{toappendix}
    \subsection{Graph Pattern Calculus}
    \label{apx:GPC}

    We make a brief summary on the syntax and semantics of GPC~\cite{10.1145/3584372.3588662}, focusing only on the concepts we need to formally define our property graph transformation rules.

    The \emph{atomic} GPC patterns are node and edge patterns.
    A \emph{node pattern} is of the form $\left(x : \ell \right)$  and an \emph{edge pattern} is of the form $\arr{double}{\mathit{x} : \ell}{}$.
    In both cases $x$ is an optional variable (picked from a countably infinite set $\mathcal{X}$ of variables) which bounds, if present, to the matched element and $\ell$ is an optional label indicating that we want to restrict to $\ell$-elements.
    In an edge pattern $\arr{double}{}{}$ may indicate one of the two possible directions: forward $\arr{->}{}{}$ and  backward $\arr{<-}{}{}$. 
    A GPC \emph{pattern} denoted $\pi$ is inductively constructed on top of the atomic patterns by using arbitrarily many \emph{union} ($\pi + \pi$), 
    \emph{concatenation} ($\pi \cdot \pi$), \emph{conditioning} ($\pi_{\langle \theta \rangle}$), and \emph{repetition} ($\pi^{n..m}$) constructs.

    A GPC \emph{query} is of the form $\rho \, \pi$ 
    with $\rho$ a \emph{restrictor} among the set of $\mathsf{simple}$, $\mathsf{trail}$, $\mathsf{shortest}$,
    which purpose is to ensure a finite result set.
    $\mathsf{simple}$ prevents repetition of nodes along a path,
    $\mathsf{trail}$ prevents repetition of edges and $\mathsf{shortest}$
    selects only the paths of minimal length among all the paths between two nodes.

    The structure of a GPC query can be inspected using a \emph{type system}, a set of typing rules~\cite{10.1145/3584372.3588662}.
    A query is \emph{well-typed} if the rules permit to deduce a unique type to every variable appearing in the query.
    When an expression $Q$ is well-typed, the schema $\mathsf{sch}(Q)$ of this expression associates a type to each variable.

    The \emph{answer} of a GPC query $Q(\bar{x})$ on a property graph $G$, denoted $\llbracket Q \rrbracket_{G}^{\bar{x}}$ is a set of assignments.
    An \emph{assignment} binds the variables $\bar{x}$, present in the query, to values.
    \emph{Values} to be associated to variables are dependent upon the deduced type of the variable for that query.
    Hence, for each type $\tau$, there is a set of values $\mathcal{V}_\tau$.
    Values may be references to elements in the graph, e.g., for $\mathsf{Node}$ and $\mathsf{Edge}$ types.

    All answers to queries we need to define our property graph transformations will have assignments of the variables among the types $\mathsf{Node}$ and $\mathsf{Edge}$.
    In our framework, we will use GPC queries extended with the capability to use conditioning on top of joins. 
    This is not part of the specification in~\cite{10.1145/3584372.3588662}, but this is planned to be in GQL~\cite{francis_researchers_2023}.
\end{toappendix}

In the following, we introduce GPC by means of examples.
The reader can refer to~\cite{10.1145/3584372.3588662} for more details on GPC, and to~\cite{francis_researchers_2023} for insight on how it will actually be used in GQL.

In Example~\ref{ex:motivating-example}, the user can retrieve from the property graph in Figure~\ref{fig:motivating-scenario}~(\subref{re:output}) ``all people who live in the same city as a person named $\$name$'' using the following GPC query:
\begin{equation*}
        \underset{\langle name = \$name \rangle}{ \left(\: : \mathsf{Person}\right) } \: \underset{}{ \arr{->}{ \mathit{} \, : \, \mathsf{HasAddress}}{}  } \:
        \underset{}{ \left(\: : \mathsf{City}\right) } \: \underset{}{ \arr{<-}{ \mathit{} \, : \, \mathsf{HasAddress}}{}  } \:
        \underset{}{ \left(y : \mathsf{Person}\right) }
\end{equation*}
This is an example of a \emph{path pattern}, which is essentially a regular path query~\cite{10.1145/2463664.2465216} augmented with \emph{conditioning}: the filter $\langle name = \$name \rangle$ checks that the value of the property $name$ is indeed the one sought. 
Given a property graph, this pattern returns the nodes that can be matched to $y$.

\looseness=-1
One can also use \emph{graph patterns} in GPC (also called patterns or queries in this paper), which are conjunctions of path patterns. For example, the following query retrieves pairs of people living in the same city, such that one person knows, possibly indirectly, the other one:
\begin{equation*}
    \begin{split}
        &\underset{}{ \left(x : \mathsf{Person}\right) } \: \underset{}{ \arr{->}{ \mathit{} \, : \, \mathsf{HasAddress}}{}  } \:
        \underset{}{ \left(\: : \mathsf{City}\right) } \: \underset{}{ \arr{<-}{ \mathit{} \, : \, \mathsf{HasAddress}}{}  } \:
        \underset{}{ \left(y : \mathsf{Person}\right) }, \\
        &\underset{}{ \left(x\right) }  \: \underset{}{ \arr{->}{ \mathit{} \, : \, \mathsf{Knows}}{}  }^{1..\infty} \: 
        \underset{}{ \left(y\right) } \,.
    \end{split}
\end{equation*}
Such graph patterns generalize conjunctive two-way regular path queries~\cite{10.1145/2463664.2465216} to property graphs.

\looseness=-1
In GPC, each path pattern occurring in a graph pattern must be qualified with a \emph{restrictor} among the set of $\mathsf{simple}$, $\mathsf{trail}$ (used by default if none is given) and $\mathsf{shortest}$.
The restrictor's purpose is to ensure a finite result set: $\mathsf{simple}$ prevents repetition of nodes along a path;
$\mathsf{trail}$ prevents repetition of edges; and $\mathsf{shortest}$ selects only the paths of minimal length among all the paths between two nodes.

For the ease of exposition, we simplify the semantics of GPC. We assume that a pattern only returns a set of bindings
(in~\cite{10.1145/3584372.3588662}, a tuple of witnessing paths is also returned with each binding).  In GPC, variables used in the scope of a repetition operator, such as ${1..\infty}$, are called \emph{group variables} and are bound to lists of nodes or edges. The remaining variables are called \emph{singleton variables} and are bound to single nodes or edges. 
For the purpose of our transformation formalism we restrict the output of queries to singleton variables. 

For a GPC pattern $P$, a tuple $\bar x$ of singleton variables in $P$, and a property graph $G$, we write $\llbracket P \rrbracket ^{\bar x}_G$ for the set of bindings of $\bar x$ returned by $P$ on $G$.
For instance, if $P$ is the first query above and $G$ is the property graph  depicted in Figure~\ref{fig:motivating-scenario}~(\subref{re:output}), we have 
$ \llbracket P \rrbracket^y_G = \{ (y \mapsto \mathsf{p2}) \}$ when $\$name$ is \textit{``Jean''},
$ \llbracket P \rrbracket^y_G = \{ (y \mapsto \mathsf{p1}) \}$ when $\$name$ is \textit{``Robert''}, and $ \llbracket P \rrbracket^y_G = \emptyset$ for any other name. 
(Note that the $\mathsf{trail}$ restrictor has been used by default.)

\section{Property graph transformations}
\label{sec:pgt}

In this section, we present our declarative formalism for specifying property graph transformations. An example is given in Figure~\ref{fig:T}. The specification consists of two rules. Each rule collects data from the input graph with a GPC pattern on the left of $\Longrightarrow$, and specifies elements of the input graph using the expression on the right. This expression resembles a GPC pattern, but it has specifications of the element's property values instead of filters and specifications of element identifiers instead of variables to be matched (new variables will reappear on the right-hand side, in a slightly different role).
In what follows we discuss how new identifiers are generated using \emph{Skolem functions} (Section~\ref{subsec:skolem}) and how identifiers, labels, and properties of output elements are specified using \emph{content constructors} (Section~\ref{subsec:content}). Then, we describe the general form of rules (Section~\ref{subsec:rules}) and  explain their semantics in terms of a procedure that generates an output property graph given an input property graph (Section~\ref{subsec:semantics}). We shall also see if the transformation in Figure~\ref{fig:T} fixes the issues discussed in Example~\ref{ex:motivating-example}.

\begin{toappendix}
    \section{Property graph transformations}
    \label{apx:semantics}
    We provide additional notation and definitions that are used in the main proofs of this paper.
\end{toappendix}

\subsection{Generating output identifiers}
\label{subsec:skolem}

Throughout the paper we assume that all identifiers in input property graphs come from a countable set $\mathcal{S} \subset \mathcal{O}$ of \emph{input identifiers}, and ensure that all identifiers in output property graphs come from a countable set $\mathcal{T} \subset \mathcal{O} \setminus \mathcal{S}$ of \emph{output identifiers}.
Following \cite{10.1145/3584372.3588654}, to generate identifiers in the output graph, we use Skolem functions.
Specifically, we use a fixed injective Skolem function 

$$ f : \bigcup_{k \in \mathbb{N}} \left( \mathcal{O} \cup \mathsf{Const} \cup \mathcal{L} \right)^k \rightarrow \mathcal{T}. $$

In the context of relational schema mappings and data exchange, Skolem functions are used for \emph{value invention}~\cite{10.1145/2463676.2465311}, e.g., to generate artificial primary keys of new tuples in a way that makes is possible to refer to them in foreign keys. The way we use Skolem functions is similar, but not the same, because element identifiers are not data values. Rather, they are the property-graph analogue of object identifiers from the object-oriented data model~\citep{10.1145/290179.290182, 10.5555/645916.671975}. Most of the time they are invisible to the user, and are not expected to carry any information beyond the identity of the element.  
Thus, the specific choice of  function $f$ is truly irrelevant, as long as $f$ is injective. 

\begin{example}
\label{ex:motivation}
In the rules in Figure~\ref{fig:T} the Skolem function is kept implicit, but its arguments are explicitly listed. For example, in the subexpression 
$((u):\mathsf{Person})$, on the right-hand side of both rules, $(u)$ indicates that the identifier of the ouptut node is $f(u)$ where $u$ is (the identifier of) a node selected from the input property graph by the left-hand side GPC pattern, such as $\mathsf{u1}$. Because 
the same nodes $u$ are selected in both rules, the subexpressions $((u):\mathsf{Person})$ in both rules will be referring to the same output nodes. Further, $(\ell.countryName)$ specifies the node identifier as $f(\ell.countryName)$, where $\ell.countryName$ name refers to the value of the property   $countryName$ in a node $\ell$ selected from the input graph, such as \textit{``United States''}, and similarly for $(a.cityName)$. If $\ell.countryName = a.cityName$ for some $\ell$ and $a$, which can happen in our example, $(\ell.countryName)$ and $(a.cityName)$ will indicate the same output node. \hfill $\blacktriangleleft$
\end{example}

\begin{toappendix}
    \subsection{Generating output identifiers}
    
    Given a tuple of $m$ variables $\bar{x} = (x_1, \dots, x_m)$, we define the sets of \emph{value arguments} and \emph{arguments} for $\bar{x}$ as 
        $$ \mathcal{V}_{\bar{x}} \Coloneqq c \mid x_i.a $$
        $$ \mathcal{A}_{\bar{x}} \Coloneqq x_i  \mid c \mid \ell \mid x_i.a $$
    where $x_i \in \bar{x}$, $c \in \mathsf{Const}$, $\ell \in \mathcal{L}$ and $a \in \mathcal{K}$.
    We denote by $\mathcal{A}^k_{\bar{x}}$ the set of all tuples of arguments for $\bar{x}$ of length $k$. 

    For a given property graph $G$, $A = (a_1, \dots, a_k)$ a tuple of arguments for $\bar{x}$ defines a function $\mathcal{O}^m \to \left( \mathcal{O} \cup \mathsf{Const \cup \mathcal{L}} \right)^k$ defined as $(o_1, \dots, o_m) \mapsto (v_1, \dots v_k) $ where, for all $1 \leq i \leq k$:
    \begin{itemize}
        \item $v_i \coloneqq o_j$ if $a_i = x_j$;
        \item $v_i \coloneqq c$ if $a_i = c$;
        \item $v_i \coloneqq \ell$ if $a_i = \ell$;
        \item $v_i \coloneqq c$ if $a_i = x_j.a$ and $\delta_G(o_j, a) = c$.
    \end{itemize}
\end{toappendix}

\subsection{Content constructors}
\label{subsec:content}

A property graph transformation must be able to specify not only the identifiers of output elements, but also their labels and properties. For this purpose, we use content constructors.  
A \emph{content constructor} is an expression of the form:
\begin{align*}
C(\bar{x}) & \coloneqq \{ \\
    \mbox{Id:} & \: (a_1, \dots, a_k) \\
    \mbox{Labels:} & \: L \\
    \mbox{Properties:} & \: \langle k_1 = v_1, \dots, k_n = v_n \rangle \: \}
\end{align*}
where $\bar{x}$ is a tuple of variables, $L$ is a finite set of labels; each $k_i$ is a property name from $\mathcal{K}$; 
each $v_i$ is either a data value $c \in \mathsf{Const}$, or an expression of the form $x.a$ for $x \in \bar{x}$ and $a \in \mathcal{K}$; 
and each $a_i$ is either a constant $c \in \mathsf{Const}$, or a label $\ell \in \mathcal{L}$, or an expression of the form $x.a$ or $x$ for $x \in \bar{x}$ and $a \in \mathcal{K}$. The field $\mbox{Id}$ specifies the identity of the node by listing the arguments to be fed to the Skolem function. The fields $\mbox{Labels}$ and $\mbox{Properties}$ specify labels and properties present in an element. Importantly, they do not forbid additional labels and properties, which will allow the user to split the description of an element across multiple rules, if the user so desires. We write $C.\mbox{Id}$ for the content of the $\mbox{Id}$ field of $C$, and similarly for other fields.
When $\bar{x}$ is clear from the context, we simply write $C$ instead of $C(\bar{x})$.

\begin{example} 
\label{ex:content} 
In the first rule in Figure~\ref{fig:T}, new $\mathsf{Country}$ nodes are described using the following content constructor: 
    \begin{align*}
        C_{\mathrm{t}}(a, u, \ell) & \coloneqq \{ \\
        \mbox{Id:} & \: (\ell.countryName) \\
        \mbox{Labels:} & \: \{ \mathrm{Country} \} \\
        \mbox{Properties:} & \: \langle name = \ell.countryName, \: code = \ell.countryCode \rangle \: \}.
    \end{align*}
It specifies the identities and the values of properties $name$ and $code$ of new $\mathsf{Country}$ nodes in terms of the values of properties $countryName$ and $countryCode$ retrieved from elements to which variable $\ell$ is bound in the input graph. 
Rather then using the abstract syntax introduced above, the rule in Figure~\ref{fig:T} presents $C_{\mathrm{t}}$ in GPC-like syntax~\cite{10.1145/3584372.3588662} as
    \begin{equation*}
        \underset{\langle name = \ell.countryName, \: code = \ell.countryCode \rangle}{\left((\ell.countryName) : \mathsf{Country}\right) } \,.
    \end{equation*}
\end{example}

\subsection{Transformations}
\label{subsec:rules}

We describe transformations in terms of property graph transformation rules. Each rule brings together the data retrieved from the input property graph by a GPC pattern and a description of output elements expressed with content constructors.

We recall that the semantics of GPC is defined such that a query returns \emph{tuples}.
Each tuple represents a \emph{binding} of singleton variables in that query to elements of the property graph.

We have two kinds of \emph{property graph transformation rules}: node rules and edge rules.
A \emph{node rule} is an expression of the form:
$$ P(\bar{x}) \implies (C(\bar x)) $$
where $P(\bar{x})$ is a GPC query with singleton variables $\bar{x}$ and $C(\bar{x})$ is a content constructor. 
An \emph{edge rule} is an expression of the form:
$$ P(\bar{x}) \implies (C_{\mathrm{s}}(\bar x)) \: \underset{}{ \arr{->}{ C(\bar x)}{} } \:
(C_{\mathrm{t}}(\bar x))  $$
where $P(\bar{x})$ is a GPC query with singleton variables $\bar{x}$ and $C_{\mathrm{s}}(\bar{x}), C(\bar{x})$ and $C_{\mathrm{t}}(\bar{x})$ are content constructors.
Finally, a \emph{property graph transformation} is a finite set of property graph transformation rules.

\begin{example}
    \label{ex:edge-rule}
    The first edge rule in Figure~\ref{fig:T} is built from the content constructor $C_{\mathrm{t}}$ as defined in Example~\ref{ex:content}, and of the following two content constructors $C_{\mathrm{s}}$ and $C$:
    
    \noindent\begin{minipage}{.5\linewidth}
        \centering
        \begin{align*}
            C_{\mathrm{s}}(a, u, \ell) & \coloneqq \{ \\
            \mbox{Id:} & \: (u) \\
            \mbox{Labels:} & \: \{ \mathrm{Person} \} \\
            \mbox{Properties:} & \: \langle name = u.name \rangle \: \},
        \end{align*}
    \end{minipage}%
    \hfill
    \begin{minipage}{.45\linewidth}
        \centering
        \begin{align*}
            C(a, u, \ell) & \coloneqq \{ \\
            \mbox{Id:} & \: () \\
            \mbox{Labels:} & \: \{ \mathrm{HasLocation} \} \\
            \mbox{Properties:} & \: \langle  \rangle \: \}. \qquad \qquad  \blacktriangleleft
        \end{align*} 
    \end{minipage}%
\end{example}

The above definition allows specifying multiple labels with a single constructor as well as specifying the labels of a single element using multiple rules.
This feature, illustrated in the following example, is crucial.
Without it, in the presence of type hierarchies, one would need negation in the query language to avoid duplicating output elements. 
In our setting, GPC does not permit negating patterns and it is unlikely for the complexity upper bounds in Section~\ref{sec:static} to hold when this form of negation is added.

\begin{example} 
\label{ex:running-example}
    As discussed in Example~\ref{ex:motivation}, if for some nodes $\ell$ and $a$ selected by the GPC patterns in the rules of Figure~\ref{fig:T}, $\ell.\textit{countryName}$ and $a.\textit{cityName}$ are equal, then the $C_{\mathrm{t}}$ constructors in both rules refer to the same output node. For instance for, $\ell=\mathsf{l1}$ and $a=\mathsf{a1}$ in the input graph in Figure~\ref{fig:motivating-scenario}~(\subref{re:input}), both rules refer to the node $f(``\mathit{Luxemburg}")$ in the output graph in Figure~\ref{fig:outT}. 
    In consequence, this node has two labels, $\mathsf{City}$ and $\mathsf{Country}$. This, quite likely, is not what the user actually wants. We will later see how to fix it by adjusting the rules.
    \hfill $\blacktriangleleft$
\end{example}

Property graphs are \emph{multigraphs} and our rules allow 
specifying multiple edges with the same endpoints by using different
arguments for the Skolem function. We will see an example in 
Section~\ref{sec:translation}.

We refer to the right-hand side expressions in node (resp. edge) rules as
node (resp. edge) constructors. We also allow rules of a more general form, illustrated in
Figure~\ref{fig:refined}, where a comma-separated list of node and edge constructors can be used on the right-hand side.
We also support aliasing, with scope limited to a single rule. For instance, in the rule in Figure~\ref{fig:refined}, we introduce alias $x=(u)$ in the
first edge constructor, and use it in the second edge contructor. Both
these extensions are syntactic sugar. To eliminate aliases, we
simply substitute them with their definitions: in the example, we
replace $x$ in the second edge constructor with $(u)$. Then, we split the rules: for each node or edge constructor on the right-hand
side, we create a separate rule with the same GPC pattern
on the left.

\subsection{Semantics}
\label{subsec:semantics}

\begin{toappendix}
    \subsection{Semantics}
\end{toappendix}

In this section, we describe operationally in Algorithm~\ref{alg:exec} how a transformation given as a set of node and edge rules 
turns an input property graph into an output property graph.
In Section~\ref{sec:translation}, we will see how to implement this efficiently in an existing graph database.

Given a GPC query $ P(\bar{x})$, a content constructor $C(\bar x)$ and a binding $\bar o$ for $P(\bar{x})$ over an input property graph $G$, we define $C.\mbox{Id}(\bar o)$ by 
replacing in $C.\mbox{Id}$ each $x_j$ with $o_j$ and each $x_j.a$ with $\delta_G(o_j, a)$. Similarly, we define $C.\mbox{Properties}(\bar o)$ by replacing in $C.\mbox{Properties}$ each $x_j.a$ with $\delta_G(o_j, a)$.

\begin{algorithm}
\caption{Semantics of a set of transformation rules.}
\label{alg:exec}
\begin{algorithmic}[1]
\Require A property graph $G$ and a set of transformation rules $T$.
\Ensure An \emph{output of the transformation $T$ over $G$}, a property graph $T(G) = \langle N, E, \lambda, \mathsf{src}, \mathsf{tgt}, \delta \rangle$.
\vspace{2pt}
\State{initialize $T(G)$ to the empty property graph }
\vspace{-2pt}
\For{each edge rule $ P(\bar{x}) \implies (C_{\mathrm{s}}(\bar x)) \: \underset{}{ \arr{->}{ C(\bar x)}{} } \:
(C_{\mathrm{t}}(\bar x)) \in T$}
    \vspace{-7pt}
    \State add rules $P(\bar{x}) \implies (C_{\mathrm{s}}(\bar x))$ and  $P(\bar{x}) \implies (C_{\mathrm{t}}(\bar x))$ to $T$ \label{line:split}
\EndFor
\For{each node rule $P(\bar{x}) \implies (C(\bar x))  \in T$} \label{line:loop-node} \label{line:on}
    \For{each binding $\bar{o} \in \llbracket P\rrbracket ^{\bar{x}}_G$} \label{line:oon}
        \State{$N \leftarrow N\cup \{ o \coloneqq f(C.\mbox{Id}(\bar{o})) \}$} \label{line:id}
        \State{$\lambda(o) \leftarrow \lambda(o) \cup C.\mathrm{Labels}$} \label{line:labels}
        \State{set $\delta(o, k)$ to $c$ if $C.\mathrm{Properties}(\bar o)$ sets property $k$ to $c$} \label{line:record}
    \EndFor
\EndFor
\vspace{-4pt}
\For{each edge rule $ P(\bar{x}) \implies (C_{\mathrm{s}}(\bar x)) \: \underset{}{ \arr{->}{ C(\bar x)}{} } \:
(C_{\mathrm{t}}(\bar x)) \in T$} \label{line:oe}
    \vspace{-7pt}
    \For{each binding $\bar{o} \in \llbracket P\rrbracket ^{\bar{x}}_G$} \label{line:ooe}
        \State{$o_{\mathrm{s}} \leftarrow f(C_{\mathrm{s}}.\mbox{Id}(\bar{o})); \ o_{\mathrm{t}} \leftarrow f(C_{\mathrm{t}}.\mbox{Id}(\bar{o}))$} \label{line:Enodes}
        \State{$E \leftarrow E \cup \{ o \coloneqq f(o_{\mathrm{s}},C.\mbox{Id}(\bar{o}),o_{\mathrm{t}}) \}$} \label{line:Eid}
        \State{$\mathsf{src}(o) \leftarrow o_{\mathrm{s}}; \ \mathsf{tgt}(o) \leftarrow o_{\mathrm{t}}$} \label{line:Eendpoints}
        \State{$\lambda(o) \leftarrow \lambda(o) \cup C.\mathrm{Labels}$} \label{line:Elabels}
        \State{set $\delta(o, k)$ to $c$ if $C.\mathrm{Properties(\bar o)}$ sets property $k$ to $c$} \label{line:Erecord}
    \EndFor
\EndFor
\end{algorithmic}
\end{algorithm}

\begin{figure*}[t]
    \begin{subfigure}[b]{0.95\textwidth}
        \centering\small
        \captionsetup{justification=centering}
        \begin{flalign}
            \underset{\langle u.address = a.aid, \: u.address = \ell.aid \rangle}{ 
                 (u : \mathsf{User}), 
                 (a : \mathsf{Address}), 
                 (\ell : \mathsf{Location})
            } \implies 
            \underset{\langle name = u.name \rangle}{ \left((u) : \mathsf{Person}\right) } & \: \underset{}{ \arr{->}{ \mathit{} \, : \, \mathsf{HasLocation}}{}  } \:
            \underset{\langle name = \ell.countryName, \: code = \ell.countryCode \rangle}{ \left((\ell.countryName) : \mathsf{Country}\right) } \label{rule:country} \\ 
            \underset{\langle u.address = a.aid, \: u.address = \ell.aid \rangle}{ 
                (u : \mathsf{User}), 
                (a : \mathsf{Address}), 
                (\ell : \mathsf{Location})
             } \implies
             \underset{\langle name = u.name \rangle}{ \left((u) : \mathsf{Person}\right) } & \: \underset{}{ \arr{->}{ \mathit{} \, : \, \mathsf{HasAddress}}{}  } \:
            \underset{\langle name = a.cityName, \: code = a.cityCode \rangle}{ \left((a.cityName) : \mathsf{City}\right) } \label{rule:city}
        \end{flalign}        
    \end{subfigure}
    \vspace{-0.2em}
    \caption{Transformation $T$ given as a set of rules.}
    \vspace{-0.2em}
    \label{fig:T}
\end{figure*}
\begin{figure}
    \centering
    \includegraphics[width=6.5cm]{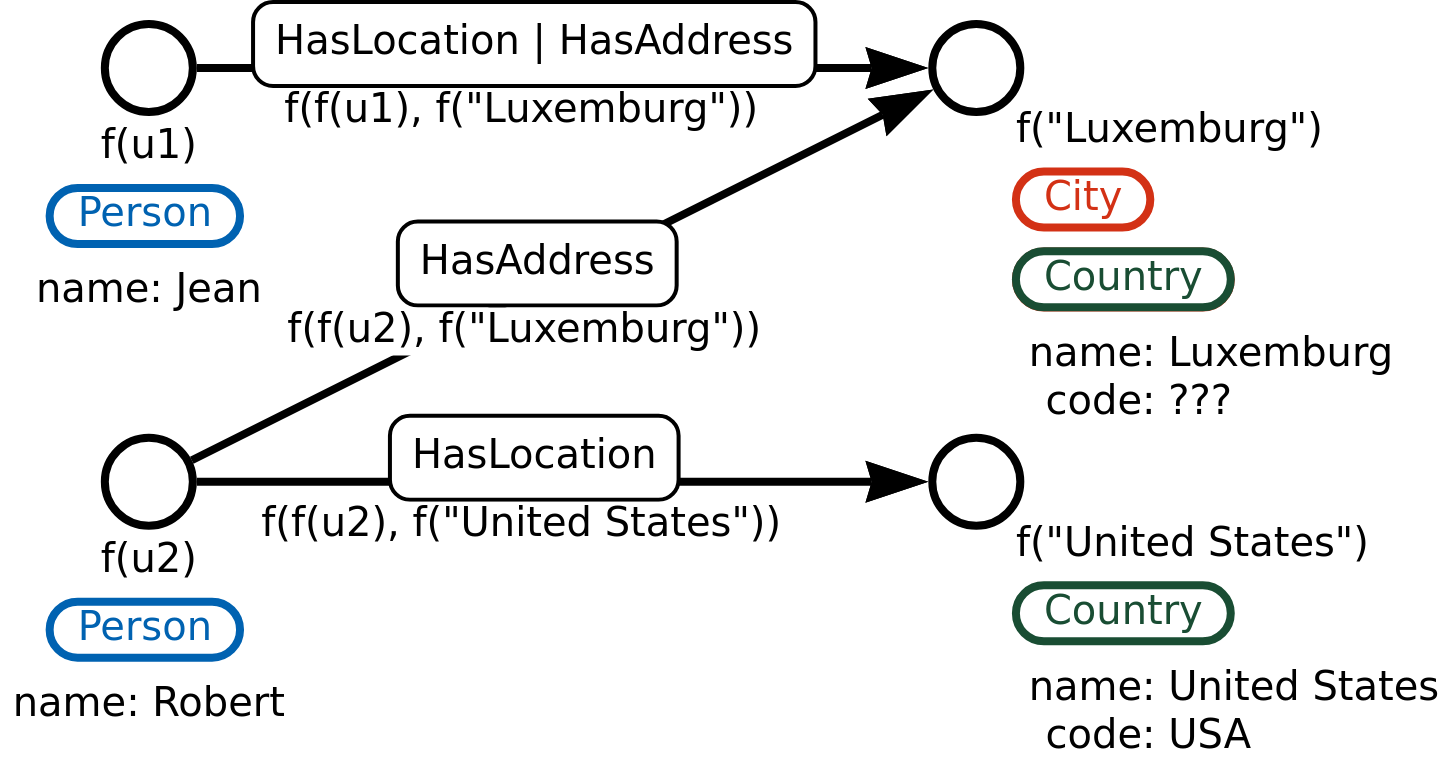}
    \caption{Output property graph $T(G)$.}
    \vspace{-0.5em}
    \label{fig:outT}
\end{figure}

\begin{example}
    \label{ex:alg}
    We describe, step by step, the operations carried out by Algorithm~\ref{alg:exec} on the input consisting of the property graph $G$ from Figure~\ref{fig:motivating-scenario} (\subref{re:input})
    and the transformation $T_1$ which contains only the first of the two rules in Figure~\ref{fig:T}.

    First, the GPC query $$P(u, a, \ell) \coloneqq \underset{\langle u.address = a.aid, \: u.address = \ell.aid \rangle}{ 
                        (u : \mathsf{User}), 
                        (a : \mathsf{Address}), 
                        (\ell : \mathsf{Location})
                    }$$ is executed on $G$ (only once in the entire process) and outputs the set of bindings 
    $\llbracket P \rrbracket^{u, a, \ell}_G = \{ (u \mapsto \mathsf{u1}, a \mapsto \mathsf{a1}, \ell \mapsto \mathsf{l1}), \: (u \mapsto \mathsf{u2}, a \mapsto \mathsf{a2}, \ell \mapsto \mathsf{l2}) \}.$

    In Line~\ref{line:split}, the single edge rule of $T_1$ is split into node rules $P(\bar{x}) \implies (C_{\mathrm{s}}(\bar x))$ and  $P(\bar{x}) \implies (C_{\mathrm{t}}(\bar x))$, 
    where $C_{\mathrm{s}}$ and $C_{\mathrm{t}}$ have been defined in Example~\ref{ex:edge-rule} and ~\ref{ex:content}, respectively.
    These two node rules are added to $T_1$, which initially contains no node rules.

    Suppose that the node rule $P(\bar{x}) \implies (C_{\mathrm{s}}(\bar x))$ is considered first in the loop in Line~\ref{line:loop-node}.
    Two output nodes are created  with respective identifiers $f(u1)$ and $f(u2)$ (Line~\ref{line:id}), one for each binding. Initially, they have no labels, $\lambda(f(u1))=\lambda(f(u2))=\emptyset$, and no properties. Then  both get label $\mathsf{Person}$ (Line~\ref{line:labels}) 
    and their property $name$ is set to \textit{``Jean''} and \textit{``Robert''}, respectively (Line~\ref{line:record}).

    Next, the algorithm moves to the  node rule $P(\bar{x}) \implies (C_{\mathrm{t}}(\bar x))$.
    Two nodes are created in the output with respective identifiers $f(``\mathit{Luxemburg}")$ and $f(``\mathit{United\ States}")$ (Line~\ref{line:id}), one for each binding; 
    they both get label $\mathsf{Country}$ (Line~\ref{line:labels}); 
    and their properties $name$ and $code$ are filled in (Line~\ref{line:record}).

    \tolerance=1000 \looseness=-1
    Finally, the algorithm steps through the only edge rule in $T_1$.
    For the first binding, the nodes corresponding to the endpoints of the edge that has to be created, namely $o_{\mathrm{s}} \coloneqq f(u1)$ and $o_{\mathrm{t}} \coloneqq$ $f(``\mathit{Luxemburg}")$ are retrieved (Line~\ref{line:Enodes}).
    They correspond to the nodes that were created, from this binding, by the node rules that were added to $T_1$ in Line~\ref{line:split}.
    An edge with id $f(f(u1), f(``\mathit{Luxemburg}"))$ is created (Line~\ref{line:Eid}); 
    its source and target are set to $f(u1)$ and $f(``\mathit{Luxemburg}")$, respectively (Line~\ref{line:Eendpoints}); 
    it gets label $\mathsf{HasLocation}$ (Line~\ref{line:Elabels});
    and no property is filled (Line~\ref{line:Erecord}).
    For the second binding, an edge is created by the same process between the nodes $f(u2)$ and $f(``\mathit{United\ States}")$. \hfill $\blacktriangleleft$
\end{example}

The role of Algorithm~\ref{alg:exec} is to give semantics to a set of transformation rules: it explains how the outputs of the multiple rules are consolidated into a single output property graph.
The following result shows that our transformations are indeed graph-to-graph transformations, offering a way to meet the expected requirements of future versions of standard graph query languages~\cite{francis_researchers_2023}.

\begin{propositionrep}
    \label{prop:sanity-check}
    Given an input property graph $G$ and a property graph transformation $T$, 
    Algorithm~\ref{alg:exec} always returns a valid instance of the property graph data model.
\end{propositionrep}

\begin{proof} 
    Given $T$ and $G$ with identifiers from $\mathcal{S}$, let $T(G) \coloneqq \langle N, E, \lambda, \mathsf{src}, \mathsf{tgt}, \delta \rangle$ be a property graph returned by Algorithm~\ref{alg:exec}.
    We have to check that
    (i) both $N$ and $E$ are finite and disjoint, 
    (ii) all elements have a finite number of labels, and
    (iii) every edge has exactly one source and one target.

    (i) The set of bindings resulting from querying a property graph $G$ with the query $P(\bar{x})$, $\llbracket P \rrbracket^{\bar{x}}_G$, is assumed to be finite,
    moreover, we have a finite number of rules in $T$, hence the finiteness of $N \cup E$. 
    A similar reasoning shows the finiteness of the label set for each element in $T(G)$; this is because each rule can mention at most a finite number of labels.

    (ii) We now show that $N \cap E = \emptyset$. Let us assume that $o \in \mathcal{T}$ is both a node and an edge id in $T(G)$ -- respectively resulting from a node rule $R \coloneqq P(\bar{x}) \implies \left(D\right)$ for $\bar{o}$ and 
    an edge rule $S \coloneqq Q(\bar{y}) \implies (C_{\mathrm{s}}) \: \underset{}{ \arr{->}{ C}{}  } \: (C_{\mathrm{t}})$ for $\bar{p}$.
    The \emph{injectivity} of the Skolem function $f$ enforces that $o$ has been generated, in both cases, by using the same arguments. 
    Moreover, by \emph{injectivity} again, we necessarily have $D.\mathsf{Id}(\bar{o}) = (o_{\mathrm{s}}, C.\mathsf{Id}(\bar{p}), o_{\mathrm{t}})$
    for some $o_{\mathrm{s}}, o_{\mathrm{t}} \in N$, with $C.\mbox{Id} \in \mathcal{A}^{k-2}_{\bar{x}}$. 
    (Note that $o_{\mathrm{s}}$ and $o_{\mathrm{t}}$ have been respectively obtained from the source and target rules of $S$ for $\bar{p}.)$
    By definition of the range of the Skolem function $f$, $o_{\mathrm{s}}$ and $o_{\mathrm{t}}$ belong to $\mathcal{T}$; thus, they could not be equal to the first and last values of $D.\mbox{Id}$.
    We conclude that $N \cap E = \emptyset$.

    (iii) Finally, by \emph{injectivity} of $f$, for an $o \in \mathcal{O}$ which is an edge id in $T(G)$, there are, by definition, exactly one $o_{\mathrm{s}} \in N$ and one $o_{\mathrm{t}} \in N$ which correspond to the source and the target of this edge.
\end{proof}

Although Algorithm~\ref{alg:exec} always returns a valid property graph (Proposition~\ref{prop:sanity-check}),
property values may depend on the order in which the rules and bindings are considered in Lines~\ref{line:on}--\ref{line:oon} and~\ref{line:oe}--\ref{line:ooe}.
Hence, the result of the transformation may be ill-defined on some inputs. 
We investigate this further in the next section.

\section{Detecting conflicts}
\label{sec:conflicts}

As one would expect from any expressive property graph transformation language, our formalism supports manipulating properties of output nodes and edges. Compared to purely structural mechanisms, such as~\cite{10.1145/3584372.3588654}, this poses additional challenges. 

\begin{example}
    \label{ex:conflict}
    Let us continue Example~\ref{ex:alg} by now considering the two-rule transformation $T$ presented in Figure~\ref{fig:T}.
    The second rule gets split into two nodes rules, one of which is 
    $$ P(\bar{x}) \implies \underset{\langle name = a.cityName, \: code = a.cityCode \rangle}{ \left((a.cityName) : \mathsf{City}\right) }.$$
    Suppose that this node rule is processed in Line~\ref{line:loop-node} after the two node rules discussed in Example~\ref{ex:alg}. The algorithm attempts twice to create a node with identifier  
    $f(``\mathit{Luxemburg}")$ (Line~\ref{line:id}), once for each binding.
    However, a node with identifier $f(``\mathit{Luxemburg}")$ has been already created by the second node rule in Example~\ref{ex:alg}. In consequence, the label $\mathsf{City}$ is added to this node (Line~\ref{line:labels}) and its properties $name$ and $code$ are set to $Luxemburg$ and  $1457$, respectively (Line~\ref{line:record}), overriding previous values $Luxemburg$ and  $LUX$. This means that one of the two values of property $code$ is lost and it depends on the processing order of rules which one it is. Indeed, the mapping now conflates not only two cities called Luxemburg, as in Example~\ref{ex:motivating-example}, but also the country  Luxemburg. This time, however, the error is easy to spot: looking at the rules we see immediately that the identity of the output nodes depends exclusively on the name of the city/country, which means that all cities and countries with the same name are conflated. We can fix the transformation easily by including information about the corresponding country in the identity of each $\mathsf{City}$ node, for instance by replacing $(a.cityName)$ with $(a.cityName,\ell.countryName)$ in rule (\ref{rule:city}) in Figure~\ref{fig:T}.
    \hfill $\blacktriangleleft$
\end{example}

Detecting the modelling error in the rules in Figure~\ref{fig:T} requires human insight (basic understanding of geography) but we hope to make it easier by insisting on explicit identity specification in transformations. On the other hand, setting an output property to conflicting values is something one can try to capture abstractly and detect automatically. This is what we do next. In the reminder of this section we focus on detecting conflicts statically, by analysing a set of transformation rules to check if it can exhibit this pathological behavior on some input. We come back to handling conflicts dynamically in Section~\ref{sec:translation} and   Section~\ref{sec:experiments}. Due to the limited space, most proofs are moved to  the appendix~\cite{TPG-github}.

\subsection{Consistency}
\label{sec:static}

\begin{toappendix}
    \section{Detecting conflicts}
    \label{apx:static}
    We formally define the notion of \emph{conflicts} and provide the proofs for the results in Section~\ref{sec:conflicts}.
    \subsection{Consistency checking}
    \label{apx:consistency-checking}
    We now formalize the notion of \emph{node conflict}. 
    For any $k \geq 0$, let
    $$  R \coloneqq P(\bar{x}) \implies \underset{\langle \dots, a = v, \dots \rangle}{( A : L ) } \mbox{ and } S \coloneqq Q(\bar{y}) \implies \underset{\langle \dots, a = w, \dots \rangle}{( B : M) } $$
    be two node rules in $T$ with $L, M \in \mathcal{P}_{fin}(\mathcal{L})$.
    These two \emph{node} rules are \emph{potentially conflicting} on the property $a$ when:
    \begin{itemize}
        \item their respective argument lists have same length, i.e., $A = (a_1, \dots, a_k) \in \mathcal{A}^k_{\bar{x}}$ and $B = (b_1, \dots, b_k) \in \mathcal{A}^k_{\bar{y}}$ for a $k \geq 0$;
        \item their argument lists are \emph{compatible}, which means that for each $1 \leq i \leq k$, $a_i$ is a value argument if and only if $b_i$ is a value argument;
        \item if $a_i$ and $b_i$ are respectively $x_i$ and $y_j$, then $\mathsf{sch}(P)(x_i)$ should be equal to $\mathsf{sch}(P)(y_j)$;
        \item they have \emph{potentially conflicting properties}, which means that they both define the same property to (possibly) different values, 
        i.e., $a \in \mathcal{K}, v \in \mathcal{V}_{\bar{x}}$ and $w \in \mathcal{V}_{\bar{y}}$.
    \end{itemize}
    A \emph{node conflict} for a pair of possibly conflicting rules on the property $a$ \emph{occurs} in a property graph $G$ whenever it exists $\bar{o} \in \llbracket P \rrbracket^{\bar{x}}_G$ and $\bar{p} \in \llbracket Q \rrbracket^{\bar{y}}_G$ with $A(\bar{o}) = B(\bar{p})$ and $v(\bar{o}) \neq w(\bar{p})$.

    We now formalize the notion of \emph{edge conflict}. 
    For any $k \geq 0$, let
    $$ R \coloneqq P(\bar{x}) \implies (A_{\mathrm{s}} : \,) \underset{\langle \dots, a = v, \dots \rangle}{\left[ \arr{->}{ A \, : \, L }{} \right]} (A_{\mathrm{t}} : \,)
    \mbox{ and } S \coloneqq Q(\bar{y}) \implies (B_{\mathrm{s}} : \,) \underset{\langle \dots, a = w, \dots \rangle}{\left[ \arr{->}{ B \, : \, M }{} \right]} (B_{\mathrm{t}} : \,) $$
    be two edge rules in $T$ with $L, M \in \mathcal{P}_{fin}(\mathcal{L})$.
    These two \emph{edge} rules are \emph{potentially conflicting} on the property $a$ when:
    \begin{itemize}
        \item their respective argument lists have same length, i.e., $A = (a_1, \dots, a_k) \in \mathcal{A}^k_{\bar{x}}$ and $B = (b_1, \dots, b_k) \in \mathcal{A}^k_{\bar{y}}$ for a $k \geq 0$;
        \item their argument lists are \emph{compatible}, which means that for each $1 \leq i \leq k$, $a_i$ is a value argument if and only if $b_i$ is a value argument;
        \item if $a_i$ and $b_i$ are respectively $x_i$ and $y_j$, then $\mathsf{sch}(P)(x_i)$ should be equal to $\mathsf{sch}(P)(y_j)$;
        \item the three previous points also apply to the pairs $(A_{\mathrm{s}}, A_{\mathrm{t}})$ and $(B_{\mathrm{s}}, B_{\mathrm{t}})$;
        \item they have \emph{potentially conflicting properties}, which means that they both define the same property to (possibly) different values, 
        i.e., $a \in \mathcal{K}, v \in \mathcal{V}_{\bar{x}}$ and $w \in \mathcal{V}_{\bar{y}}$.
    \end{itemize}
    An \emph{edge conflict} for a pair of possibly conflicting rules on the property $a$ \emph{occurs} in a property graph $G$ whenever it exists $\bar{o} \in \llbracket P \rrbracket^{\bar{x}}_G$ and $\bar{p} \in \llbracket Q \rrbracket^{\bar{y}}_G$ with $A(\bar{o}) = B(\bar{p})$, $A_{\mathrm{s}}(\bar{o}) = B_{\mathrm{s}}(\bar{p})$, $A_{\mathrm{t}}(\bar{o}) = B_{\mathrm{t}}(\bar{p})$ and $v(\bar{o}) \neq w(\bar{p})$.
\end{toappendix}

By a \emph{conflict} we mean a situation when Algorithm~\ref{alg:exec} resets a previously set property to a different value, as illustrated in Example~\ref{ex:conflict}. 
A transformation $T$ is
\emph{consistent} if for every input property graph $G$, no execution of Algorithm 1 results in a conflict. Note that even a transformation consisting of a single rule can be inconsistent, because different bindings for the same rule can cause a conflict.  

We study the following fundamental static analysis problem, in the setting where there is no source schema constraining the set of possible input property graphs. 

\begin{description}[leftmargin=2em]
    \item[Consistency.] Given a transformation $T$, check if $T$ is consistent.
\end{description}

As we show next, consistency of transformations is deeply related to  satisfiability of GPC patterns. A GPC pattern is \emph{satisfiable} if it returns a non-empty set of answers on some property graph. Towards the goal of establishing complexity lower bounds for the consistency problem, we provide a polynomial-time reduction from the satisfiability problem for GPC.

\begin{description}[leftmargin=2em]
    \item[Satisfiability.] Given a GPC pattern $P$, check if $P$ is satisfiable. 
\end{description}

\begin{lemma}
    \label{lemma:reduc}
    The satisfiability problem for GPC is \textsc{PTime}-reducible to the transformation consistency problem.
\end{lemma}

\begin{proof}
For a GPC pattern $P$, let $T_P$ be the transformation consisting of the following two rules
$$ P() \implies \underset{\langle k = 5 \rangle}{( (c) : \ell ) } \quad \mbox{ and } \quad  P() \implies \underset{\langle k = 7 \rangle}{( (c) : \ell )} $$
for some fixed label $\ell$, constant $c$, and property name $k$.
These rules are not conflicting with themselves, because their node constructors do not depend on the binding. However, they are conflicting with each other on a graph $G$ if $P$ returns at least one answer on $G$. Hence, $T_P$ is consistent iff  $P$ is not satisfiable.
\end{proof}

For the converse of Lemma~\ref{lemma:reduc} to hold we need to move to GPC+, a simple extension of GPC with projection and union~\cite{10.1145/3584372.3588662}.
\begin{lemmarep}
    \label{lemma:reduc-converse}
    The transformation consistency problem is \textsc{PTime}-reducible to the satisfiability problem for GPC+.
\end{lemmarep}

\begin{proofsketch}
For a pair of rules $R$ and $S$, and an attribute $a$, we can write a Boolean GPC+ query $Q_{R, S, a}()$ that detects if some matches of $R$ and $S$ lead to different values for attribute $a$ in the same element of the output graph. Because there are polynomially many such triples, we can take the union of all such queries to obtain the final GPC+ query to be checked for satisfiability.
\end{proofsketch}

\begin{proof}
    Recall that a node conflict for a pair of possibly conflicting rules $R$ and $S$ on a property $a$ occurs on a property graph $G$ whenever it exists $\bar{o} \in \llbracket P \rrbracket^{\bar{x}}_G$ and $\bar{p} \in \llbracket Q \rrbracket^{\bar{y}}_G$ with $A(\bar{o}) = B(\bar{p})$ and $v(\bar{o}) \neq w(\bar{p})$.
    We can rewrite all of those conditions in a single boolean GPC query:
    $$Q_{(R,S,a)}() \coloneqq P(\bar{x}), Q(\bar{y}), A = B, v \neq w$$ which is satisfiable on a property graph $H$ \emph{iff} this specific conflict occurs on $H$.
    Note that we assume w.l.o.g. in the following construction that $\bar{x}$ and $\bar{y}$ are disjoint sets of variables.

    Similarly, for an edge conflict, we obtain a single boolean GPC query with the same properties:
    $$Q_{(R,S,a)}() \coloneqq P(\bar{x}), Q(\bar{y}), A = B, A_{\mathrm{s}} = B_{\mathrm{s}}, A_{\mathrm{t}} = B_{\mathrm{t}}, v \neq w$$

    We provide an example of the GPC pattern encoding $v \neq w$. 
    Let assume that $v \coloneqq x_i.b$ and $w \coloneqq y_j.c$ and $\mathsf{sch}(P)(x_i) = \mathsf{Node}$ and $\mathsf{sch}(Q)(y_j) = \mathsf{Edge}$, the following join query encodes $v \neq w$:
    $$ \underset{\langle \neg \left( x_i.b = y_j.c \right) \rangle}{ \left[ (x_i),\left( \right) \arr{->}{ \mathit{y_j}}{} \left(\right)  \right] }. $$
    Similarly, we can apply point-wise this construction and take their join to encode $A = B, A_{\mathrm{s}} = B_{\mathrm{s}}$ and $A_{\mathrm{t}} = B_{\mathrm{t}}$.

    Finally, to wrap-up the proof, it is easy to see that given a property graph transformation $T$, there are at most polynomially many such triplets $(R, S, a)$ satisfying this criteria, so we can take the union of all the $Q_{(R,S,a)}()$ as the final GPC+ query on which to check for satisfiability.
\end{proof}

We now turn to study the complexity of the satisfiability problem for GPC and GPC+. The two lemmas above will allow us to draw conclusions for the  consistency problem in Section~\ref{subsec:backtocons}. 

\subsection{The complexity of satisfiability}
\label{subsec:gpc}

\begin{toappendix}
    \subsection{GPC satisfiability} 
    \label{apx:gpc}
\end{toappendix}

In Theorem~\ref{th:gpc-sat}, we establish that checking if a GPC+ query is satisfiable is a \textsc{PSpace}-complete problem (modulo certain assumptions on the use of restrictors).
We believe that this result is interesting in its own right, beyond the application to transformation consistency we consider this paper.
Indeed, deciding whether a query expressed in a given  query language is satisfiable is a fundamental problem in database theory.
Very little is known to this date about GPC from a theoretical viewpoint, and our work is one of the first to tackle a key static analysis task related to this query language.

\begin{lemmarep}
    \label{lemma:pspace-hard}
    The satisfiability problem for GPC is \textsc{PSpace}-hard.
\end{lemmarep}

\begin{proofsketch}
    We show how to reduce the membership problem for an arbitrary \textsc{PSpace} language to the satisfiability of a GPC query. 
    Let $L$ be a language in \textsc{PSpace} and $M$ a deterministic polynomial-space Turing machine that recognizes $L$ in space $c \cdot p(n)$ for a fixed constant $c$ and polynomial $p$.
    In the following, $n$ denotes the length of the word $w$ which is an input to $M$.

    We construct the following GPC pattern $P$:
    \begin{equation*}
        \label{trips:pspace}
        P() \coloneqq \rho \:
        ( \mathit{x} )_{\langle \theta_1 \rangle }
        \left( \left[ (\mathit{u}) \, \arr{->}{}{} \, (\mathit{v}) \right]_{\langle \theta_2 \rangle} \right)^{1..\infty}
        ( \mathit{y} )_{\langle \theta_3 \rangle } \\
    \end{equation*}
    The intuition is the following. We can represent a configuration of $M$ in a single node, using a polynomial number of properties.
    The pattern $( \mathit{x} )_{\langle \theta_1 \rangle }$ is  responsible for encoding the initial configuration of $M$ over the input word $w$.
    The pattern $ \left[ (\mathit{u}) \, \arr{->}{}{} \, (\mathit{v}) \right]_{\langle \theta_2 \rangle}$ ensures that there exists a valid transition of $M$ between the configurations represented by nodes $u$ and $v$.
    Finally, $( \mathit{y} )_{\langle \theta_3 \rangle }$ specifies that node $y$ represents an accepting~configuration.

    We can use techniques similar to the proof of the Cook-Levin Theorem~\cite{10.5555/574848} to construct in time polynomial in $n$ the formul{\ae} $\theta_1$, $\theta_2$, and $\theta_3$.
    The size of $P$ is then clearly polynomial in $n$.
    This reduction works with any $\rho \in \{ \mathsf{shortest}, \mathsf{simple}, \mathsf{trail} \}$.
\end{proofsketch}

\begin{proof}
    Let $M = (Q, \Sigma, s, F, \delta)$ be the TM that recognizes $L$ in deterministic polynomial-space.
    Let $w$ be an input word of length $n$.
    Assume that $M$ works over $w$ using at most $c \cdot p(n)$ for a fixed constant $c$ and polynomial $p$ tape cells.
    
    We build a GPC query using the following set of properties $\mathcal{K}_0 \subset \mathcal{K}$ which contains all the following elements:
    \begin{itemize}
        \item $pos_{(i, \sigma)}$; the tape contains symbol $\sigma \in \Sigma$ at position $1 \leq i \leq c \cdot p(n)$;
        \item $head_{(i)}$; the head of the TM is in position  $1 \leq i \leq c \cdot p(n)$;
        \item $q$; the TM is in state $q \in Q$.
    \end{itemize}
    Notice that we will be using only two constants values: $0$ and $1$. 
    
    To encode the consistency of a state represented by the set of properties of an element $x$, we will use formula $\theta_c(x)$ defined as the conjunction of the following formulas:
    \begin{itemize}
        \item $\neg \left( x.pos_{(i, a)} = 1 \right) \vee \neg \left( x.pos_{(i, b)} = 1 \right)$;
        for $1 \leq i \leq  c \cdot p(n)$, $a, b \in \Sigma$, $a \neq b$;
        \item $\neg \left( x.head_{(i)} = 1 \right) \vee \neg \left( x.head_{(j)} = 1 \right)$;
        for $1 \leq i \neq j \leq  c \cdot p(n)$;
        \item $\neg \left( x.q = 1 \right) \vee \neg \left( x.q' = 1 \right)$;
        for $q, q' \in Q, q \neq q'$
        \item $\left( x.k = 1 \right) \vee \left(x.k = 0 \right)$; 
        for $k \in \mathcal{K}_0$.
    \end{itemize}

    $\theta_1(x)$ is a conjunction of the following formulas; it ensures that the set of properties of the element pointed to by $x$ encodes the initial configuration of the TM $M$ over $w$:
    \begin{itemize}
        \item $x.pos_{(i, a)} = 1$; 
        if $w_i = a$, for $1 \leq i \leq n$; $w$ is stored at the beginning of the tape;
        \item $x.pos_{(i, \square)} = 1$; 
        for $n < i \leq c \cdot p(n)$, where $\square$ is the blank symbol; 
        the rest of the tape is filled with blank symbols;
        \item $\left( x.head_{(1)} = 1 \right) \wedge \left( x.s = 1 \right)$; initialisation of both head and state;
        \item $\theta_c(x)$; consistency check. 
    \end{itemize}

    $\theta_2(u, v)$ checks that the configuration stored in the record of $v$ can be obtained from $u$ in a single computation step of $M$; it consists in the conjunction of the following formulas:
    \begin{itemize}
        \item $\left( u.head_{(i)} = 0 \wedge u.pos_{(i,a)} = 1 \right) \implies v.pos_{(i,a)} = 1$; 
        for $1 \leq i \leq  c \cdot p(n),$ $a \in \Sigma$; 
        the tape remains unchanged unless written by head;
        \item $\left( u.head_{(i)} = 1 \wedge u.pos_{(i,a)} = 1 \wedge u.q = 1 \right) \implies \left( v.pos_{(i,b)} = 1 \wedge v.head_{(i+d)} = 1 \wedge v.q' = 1\right)$; 
        for $1 \leq i \leq  c \cdot p(n),$ $a, b \in \Sigma$, $q,q' \in Q$, $d \in \{ -1, 0, 1 \}$, when $(q',b,d) \in \delta(q,a)$;
        there is a valid transition;
        \item $\theta_c(u) \wedge \theta_c(v)$; 
        consistency checks.
    \end{itemize}

    Finally, $\theta_3$ checks that we have reached an accepting configuration: $$\theta_3(y) \coloneqq \bigvee\limits_{q \in F} y.q = 1$$

    It is clear that the pattern $P$ -- which is constructed in polynomial time given an entry word $w$ -- is satisfiable if and only if there exists an accepting run for $M$ over $w$ using at most a polynomial amount of space, i.e., \emph{iff} $w \in L \in \textsc{PSpace}$.
\end{proof}

For completeness, we also provide the matching upper bound (under some assumptions) and obtain Theorem~\ref{th:gpc-sat} as a result.
The details of the upper-bound proof are highly technical and the claim depends heavily on the design choices made for GPC. 

\begin{theoremrep}
    \label{th:gpc-sat}
    The satisfiability problem for GPC+ queries using only the $\mathsf{simple}$ and $\mathsf{trail}$ restrictors is \textsc{PSpace}-complete.
\end{theoremrep}

\begin{proof}
    By Lemma~\ref{lemma:pspace-hard}, it is only left to prove the upper-bound.
    Let $Q$ be a GPC query. 
    We assume a mild syntactic restriction that, all edge and node patterns must mention a variable in their descriptor.
    This can be achieved by picking a fresh variable when none is specified in a descriptor.
    
    We note respectively $\mathsf{Const}_0$, $\mathcal{K}_0$ and $\mathcal{L}_0$, the set of constants, keys and labels mentioned in $Q$.
    Additionally, let $d$ be the number of occurences of conditions of the form $x.a = c$ or $x.a = y.b$ in the formula.
    We extend the set $\mathsf{Const}_0$ with $d + 1$ fresh distinct constants.
    Let also $\mathcal{X}$ and $\mathcal{Y}$ be respectively the sets of variables of type $\mathsf{Node}$ and $\mathsf{Edge}$ in the schema of $Q$.
    
    \paragraph{Preliminary remarks.} We start with some key observations:
    \begin{itemize}
        \item It is clear that we cannot always guess an answer made up of a path and an assignment because some patterns are only satisfiable by paths of exponential length:
        e.g., we can simulate a counter to count up to $2^n$ with $n$ properties and a polynomial-sized formula similar to that used in Lemma~\ref{lemma:pspace-hard}; 
        \item The concatenation of patterns implies an implicit equality over the endpoints; thus, we need to book-keep the endpoints of a pattern alongside an assignment for it;
        \item The semantics and the typing rules of GPC isolate the variables under a repetition pattern;
        this means that, in a repetition pattern, we cannot refer to externally defined variables (i.e., non-local occurrences are not permitted);
        moreover, because conditions are only defined over \emph{singleton} variables (i.e., variables of type $\mathsf{Node}$ or of type $\mathsf{Edge}$ in the schema of $Q$), we cannot refer in conditions to variables appearing under a repetition sub-pattern.
        Thus, we can assume that, if a query is satisfiable, an answer path for a pattern inside a repetition pattern can be disjoint (except for its endpoints) from the answer paths for the outside of the repetition pattern.
    \end{itemize}
    This shows that we can store in polynomial space an \emph{extended} assignment corresponding to a valid answer to a path pattern $\pi$ 
    -- by extended, we mean that we additionally store the two endpoints of the answer path in variables named $\mathsf{src}_\pi$ and $\mathsf{tgt}_\pi$ (that we assume to belong to $\mathcal{X}$); 
    additionally, $\mathsf{Group}$ and $\mathsf{Maybe}$ variables will not be tracked because they can no longer be mentioned in conditions.
    In this case, an \emph{assignment} $\eta(\cdot)$ stores for each variable the description of an element which is made of its label set (referred to by $\lambda(\eta(\cdot))$) and its record (referred to by $\delta(\eta(\cdot),\cdot)$); it does not contain an identifier for this element.

    \paragraph{Saturation procedure and consistency check.} In the body of the main algorithm we make use of a procedure called \emph{saturation}. 
    The main idea of this procedure is to propagate equality and inequality constraints between variables whenever new ones are found.
    We illustrate why this procedure is needed and how it works on the following pattern:
    $$ ( \mathit{u} )_{\langle u.a = 1 \rangle }
    \arr{->}{ \mathit{z} }{}
    ( \mathit{v} ) 
    \arr{->}{ \mathit{z} }{}
    ( \mathit{w} )_{\langle w.a = 2 \rangle } $$
    Because both occurrences of $z$ must map to the same edge, the constraints over the endpoints of $z$ enforce equalities between the variables $u$, $v$ and $w$.
    Successively, these equalities imply that $u.a = w.a$, which is in conflict with the requirements of the pattern.
    Note that if we remove the conditions $u.a = 1$ and $w.a = 2$, this pattern remains unsatisfiable w.r.t. the $\mathsf{simple}$ semantics; 
    the forced repetition of nodes in bindings of this pattern was not explicit before applying the saturation procedure because no variable was reused.
    However, reusing the variable $z$ makes the pattern explicitely unsatisfiable under $\mathsf{trail}$ semantics.

    To formalize this, we introduce the notion of an \emph{equality graph} $G$ \emph{for a query} $Q$ (\emph{or a pattern} $\pi$), which is a $2$-layer undirected edge-labeled graph with nodes in each layer that respectively belong to the sets $\mathcal{Y}$ and $\mathcal{X}$.
    There are two edge labels, $=$ and $\neq$.
    Edges can only connect nodes from the same layer.
    If $G$ is an equality graph for a pattern $\pi$, there are two distinguished nodes $\mathsf{src}_\pi, \mathsf{tgt}_\pi \in \mathcal{X}$ which are respectively called the source and the target, and abbreviated $\mathsf{src}, \mathsf{tgt}$ if clear from the context.
    This structure supports the following set of operations:
    \begin{description}
        \item[Saturation:] The first step of the \emph{saturation procedure over a pattern $\pi$} is to equate the endpoints of edge variables.
        Let $e$ be an edge variable in $G$, if $G$ contains say $x$ and $y$, two distinct node variables, which are both at the source or the target of an occurrence of $e$ in $\pi$, then add $x = y$ in $G$.
        
        In a second step, it performs the transitive closure on the $=$-edges of the graph at layer $\mathcal{X}$.
        
        After obtaining a $=$-transitive graph, we pass on the inequalities:
        if there is an $\neq$-edge between say $x$ and $y$, and if both $x' = x$ and $y' = y$ in $G$, then we add $x' \neq y'$ in $G$. 

        Finally, it does a backward step to pull-up the inequalities:
        if there is an $\neq$-edge between two elements in $\mathcal{X}$ and if both are either the source or the target of some edges in $Q$, then it adds an inequality edge between these two edge variables. 
        \item[Check consistency:] 
        Given an assignment for all the variables mentioned in the equality graph, check if all equalities are satisfied in this assignment:
        if $x$ and $y$ are two variables with $x = y$ in $G$, check if their assignments are strictly the same.
        Moreover, check if there is no conflict in the graph: a \emph{conflict} is when there are both an $=$-edge and an $\neq$-edge between the same pair of nodes of $G$; or if there is a $\neq$-loop.
        \item[Merge:] Given two equality graphs $G_1$ and $G_2$, \emph{merging both on nodes} $x$ \emph{and} $y$ \emph{for policy} $\rho$ consists in:
        \begin{enumerate}
            \item taking the union of their vertices and edges;
            \item adding an $=$-edge between $x$ and $y$; if $x$ and $y$ are given as parameters;
            \item if $\rho = \mathsf{simple}$, adding an $\neq$ edge from each node in the layer $\mathcal{X}$ of $G_1$ not $=$-connected to $\mathsf{tgt}$, to each node of the layer $\mathcal{X}$ of $G_2$ not $=$-connected to $\mathsf{src}$;
            (This is consistent with the semantics of the concatenation of $\rho_1 \cdot \rho_2$, where the target or $\rho_1$ must be equal to the source of $\rho_2$.);
            \item if $\rho = \mathsf{trail}$, adding an $\neq$ edge from each node of the layer $\mathcal{Y}$ of $G_1$, to each node of the layer $\mathcal{Y}$ of $G_2$;
            \item applying the saturation procedure.
        \end{enumerate} Note that the result of this operation is also a valid equality graph and that $x, y$ and $\rho$ are optional parameters.
    \end{description}
    The maximum number of nodes and edges in an equality graph for $Q$ is polynomial in $Q$; and the three procedures can be implemented in polynomial time.
    
    \paragraph{Inductive procedure for patterns.} 
    In the following, we describe a non-deterministic polynomial space procedure to check if a GPC pattern query $\rho \pi$ is satisfiable.
    This inductive procedure over the structure of $\pi$ succeeds if and only if $\pi$ is satisfiable under policy $\rho$. 
    It returns a pair consisting of an extended assignment and an equality graph over the variables in the assignment:
    \begin{itemize}
        \item Case $\pi \coloneqq (\mathit{x} : \mathsf{\ell})$.
        Guess a node element with a label set $L \subseteq \mathcal{L}_0$ s.t. $\ell \in L$ and a record consisting in a partial assignment from the keys in $\mathcal{K}_0$ to $\mathsf{Const}_0$.
        Return the pair consisting of the extended assignment binding $\mathsf{src}_\pi$, $\mathsf{tgt}_\pi$ and $x$ to this node; 
        and of the equality graph containing three nodes: $\mathsf{src}_\pi$, $\mathsf{tgt}_\pi$ and $x$, and two edges: $\mathsf{src}_\pi = x$ and $\mathsf{tgt}_\pi = x$.
        \item Case $\pi \coloneqq \arr{->}{ \mathit{y} \, : \, \mathsf{\ell}}{}$.
        Guess an edge element with a label set $L \subseteq \mathcal{L}_0$ s.t. $\ell \in L$ and a record consisting in a partial assignment from the keys in $\mathcal{K}_0$ to $\mathsf{Const}_0$.
        Return one of the following two possibilities:
        \begin{itemize}
            \item Return a pair consisting of the extended assignment binding $y$ to this edge, and $\mathsf{src}_\pi$ and $\mathsf{tgt}_\pi$ to two arbitrarily guessed endpoint nodes (which act as the source and target of $y$); 
            the equality graph contains these three elements and an $\neq$-edge between $\mathsf{src}_\pi$ and $\mathsf{tgt}_\pi$.
            \item (Loop; if $\rho$ is not $\mathsf{simple}$) Return a pair consisting of the extended assignment binding this edge to $y$ and $\mathsf{src}_\pi$ and $\mathsf{tgt}_\pi$ to the same arbitrarily guessed endpoint node;
            the equality graph contains these three elements and an $=$-edge between $\mathsf{src}_\pi$ and $\mathsf{tgt}_\pi$.
        \end{itemize}
        \item Case $\pi \coloneqq \pi_1 + \pi_2$.
        Guess $i \in \{1, 2\}$ and return the extended assignment and equality graph obtained by a recursive call to this procedure on $\pi_i$, after removing the variables that do not appear in $\pi_{3-i}$. 
        (This is because those variables are of type $\mathsf{Maybe}(\cdot)$ in $\pi$.)
        \item Case $\pi \coloneqq \pi_1 \pi_2$.
        Perform a recursive call to this procedure on both $\pi_1$ and $\pi_2$ 
        to obtain an extended assignment and an equality graph for $\pi_i$, $i \in \{1, 2\}$.
        Check whether they unify (i.e., check whether the assignments of $\pi_1$ and $\pi_2$  are strictly the same on their common variables).
        Then, merge the two equality graphs on $\mathsf{tgt}_{\pi_1}$ and $\mathsf{src}_{\pi_2}$ under the policy $\rho$; and check consistency w.r.t. the unified extended assignment.
        Return the pair consisting of the unified extended assignment and the merged equality graph after removing $\mathsf{tgt}_{\pi_1}$ and $\mathsf{src}_{\pi_2}$ and renaming $\mathsf{src}_{\pi_1}$ to $\mathsf{src}_{\pi}$ and $\mathsf{tgt}_{\pi_2}$ to $\mathsf{tgt}_{\pi}$.
        \item Case $\pi \coloneqq {\pi_1}_{\langle \theta \rangle}$.
        Perform a recursive call to this procedure on $\pi_1$ to obtain an extended assignment and an equality graph for $\pi_1$.
        Check the validity of $\theta$ (as defined in the \emph{Semantics of conditioned patterns} in~\cite{10.1145/3584372.3588662}) over the extended assignment of $\pi_1$ and return the same extended assignment and equality graph.
        \item Case $\pi \coloneqq \pi_1^{n .. m}$.
        Guess a length $k$ between $n$ and $m$. 
        (Note that $k$ can be assumed to be at most exponential in the size of the whole query if $m$ is $\infty$ by Lemma~\ref{lemma:inner};
        hence, it can  be written in binary using a polynomial amount of space.)
        Perform $k$ successive recursive calls to this procedure on $\pi_1$.
        Each time drop from the obtained extended assignment and equality set all but the $\mathsf{src}$ and $\mathsf{tgt}$ variables and the equality or inequality edge between them. 
        Check if they concatenate with the previous block obtained so far (i.e., perform the case $\pi \coloneqq \pi_1 \pi_2$).
        Return a pair consisting of an extended assignment binding the $\mathsf{src}$ of the very first guess to $\mathsf{src}_\pi$ and the $\mathsf{tgt}$ of the very last guess to $\mathsf{tgt}_\pi$; 
        and of the equality graph containing only the $\mathsf{src}_\pi$ and $\mathsf{tgt}_\pi$ nodes, possibly with an $=$ or an $\neq$ edge between them if should be.
    \end{itemize}

    \paragraph{Example.} We illustrate how our procedure works for repetition patterns with the following example:
    $$ 
    \pi \coloneqq \mathsf{simple} \left( \left[ (\mathit{u}) \, \arr{->}{}{} \, (\mathit{v}) \right]_{\langle \neg \left( u.a = v.a \right) \rangle} \right)^{0..\infty}
    $$
    We note $\pi_1 \coloneqq \left[ (\mathit{u}) \, \arr{->}{}{} \, (\mathit{v}) \right]_{\langle \neg \left( u.a = v.a \right) \rangle}$. 
    There are three different types of behavior depending on $k$:
    \begin{itemize}
        \item If $k = 0$, it returns a pair consisting of an extended assignment binding $\mathsf{src}_\pi$ and $\mathsf{tgt}_\pi$ to an arbitrarily guessed node element; 
        and of an equality graph containing the edge $\mathsf{src}_\pi = \mathsf{tgt}_\pi$.
        \item If $k = 1$, it returns a pair consisting of an extended assignment binding $\mathsf{src}_\pi$ and $\mathsf{tgt}_\pi$ to two node elements having a different value for $a$ (if both set); 
        and of the quality graph containing $\mathsf{src}_\pi \neq \mathsf{tgt}_\pi$.
        \item If $k \geq 2$, it returns a pair consisting of an extended assignment binding $\mathsf{src}_\pi$ and $\mathsf{tgt}_\pi$ to two arbitrary node elements; 
        and of the quality graph containing $\mathsf{src}_\pi \neq \mathsf{tgt}_\pi$.
    \end{itemize}

    \paragraph{Invariant.} Let $\pi$ be a pattern matched under a restrictor $\rho \in \{ \mathsf{simple}, \mathsf{trail} \}$.
    The pair ($\eta, G_\pi)$ is an output of the procedure consisting of an extended assignment and an equality graph for $\pi$ \emph{iff}
    there exists a property graph $P$ such that $(p, \mu) \in \llbracket\pi \rrbracket_P$\footnote{The set of answers to $\pi$ on $P$~\cite{10.1145/3584372.3588662}.} and if, for all $x$ s.t. $\mathsf{sch}(\pi)(x) \in  \{ \mathsf{Node}, \mathsf{Edge} \}$, we have:
    \begin{itemize}
        \item $\lambda(\mu(x)) = \lambda(\eta(x))$;
        \item for all key (property) $a$, $\delta(\mu(x), a) = \delta(\eta(x), a)$, or both are not defined;
        \item if $x = y$ in $G$ then $\mu(x) = \mu(y)$; similarly, if $x \neq y$ in $G$ then $\mu(x) \neq \mu(y)$.
    \end{itemize}
    In the following, we provide the key ideas for proving this invariant by induction:
    \begin{itemize}
        \item Case $\pi \coloneqq (\mathit{x} : \mathsf{\ell})$ and case $\pi \coloneqq \arr{->}{ \mathit{y} \, : \, \mathsf{\ell}}{}$ are trivial as $P$ can be obtained directly from $\eta$.
        \item Case $\pi \coloneqq \pi_1 + \pi_2$ by applying the hypothesis on either side. 
        Note that $\pi$ may contain fewer singleton variables than $\pi_i$.
        (This is because the variables of $\pi_i$ that do not appear in $\pi_{3-i}$ are of type $\mathsf{Maybe}(\cdot)$ in $\pi$.)
        \item Case $\pi \coloneqq \pi_1 \pi_2$ relies on the merge procedure to propagate the \emph{forced} equalities and inequalities.
        \item Case $\pi \coloneqq {\pi_1}_{\langle \theta \rangle}$ by noticing that we only need to check whether the condition holds over the extended assignment because conditions only apply on singleton variables. 
        (This is enforced by the \emph{Typing rules for the GPC type system} presented in Figure~2 of~\cite{10.1145/3584372.3588662}.)
        Notice that $\neg (u.a = v.a)$ does not lead to $u \neq v$ because either $u.a$ or $v.a$ may be undefined; thus, we don't need to update $G$.
        (Again, this is because of the \emph{Semantics of conditioned patterns} of~\cite{10.1145/3584372.3588662} where $\mu \models (x.a = x.b) $ \emph{iff} $\delta(\mu(x), a)$ and $\delta(\eta(y), b)$ are defined and equal.)
        \item Case $\pi \coloneqq \pi_1^{n .. m}$ because $\pi$ does not contain any singleton variables. 
        Hence, only a potential equality or inequality between its endpoints is tracked throughout the $k$ iteration steps over $\pi_1$.
    \end{itemize}

    \paragraph{Extension to queries (with join and conditioning).} Let $Q \coloneqq Q_1, Q_2$ be a join query with $Q_1$ and $Q_2$ matched under restrictors $\mathsf{simple}$ or $\mathsf{trail}$; and let the pairs $(\eta_i, G_{Q_i}), i \in \{ 1, 2\}$ be returned by the previous procedure on each $Q_i$.
    The procedure (for $Q$) checks if the $\eta_i$'s unify (i.e., check whether the $\eta_i$'s agree on their common variables), and if the result of merging the $G_i$'s remains consistent.

    Note that this procedure supports \emph{conditioning over join queries} (i.e., we add the following query expression $Q \coloneqq Q_{\langle \theta \rangle}$) by checking if the condition is valid over $\eta$, similar to what is done for conditioning in patterns.

    \paragraph{Extension to GPC+.} Simply guess a GPC query among all the disjuncts, and check its satisfiability using the previous procedure.

    \paragraph{Shortest restrictor.} Note that the invariant is not valid if the $\mathsf{shortest}$ restrictor is used.
    For instance, consider the following pattern, which is not supported by the above procedure:
    \begin{equation*} 
    \begin{split}
    &\mathsf{shortest} \, \mathsf{simple} \, (\mathit{x}) \left[ (\,) \, \arr{->}{}{}  \underset{\langle a = 1 \rangle}{(\mathit{u})}  \arr{->}{}{} \, (\,) \, \mbox{\LARGE +} \,
    (\,) \, \arr{->}{}{} \, (\,) \, \arr{->}{}{} \, (\mathit{u}) \, \arr{->}{}{} \, (\,) \, \arr{->}{}{} \, (\,) \right] (\mathit{y}), \\
    &\mathsf{simple} \, (x) \, \arr{->}{}{}  \underset{\langle a = 1 \rangle}{(\mathit{w})}  \arr{->}{}{} \, (y), \\
    &\mathsf{simple} \, \underset{\langle a = 2 \rangle}{(\mathit{u})}
    \end{split}
    \end{equation*}
    Nevertheless, we can easily prove that the above procedure works when all patterns in a GPC query $Q$ that are matched under the $\mathsf{shortest}$ restrictor have only their endpoints for singleton variables.
\end{proof}

\begin{toappendix}
    \begin{lemma}
        \label{lemma:inner}
        Let $\pi \coloneqq \pi_1^{\infty}$ be a sub-pattern of a GPC query $\rho\pi_0$.
        If there is an answer path for $\pi$, then there is another answer path for $\pi$ consisting of at most $k$ repetitions of $\pi_1$, with $k$ exponential in the size of $\pi_0$, and with all answer paths for $\pi_1$ being inner disjoints.
    \end{lemma}

    \begin{proof}
        There are at most exponentially many distinct records for nodes with values in $\mathsf{Const}_0$, keys in $\mathcal{K}_0$ and labels in $\mathcal{L}_0$.
        Given the \emph{Semantics of repeated patterns} in~\cite{10.1145/3584372.3588662}, only the target node of an iteration has an impact over the next iteration by being its source; only this information is transferred across successive repetitions of $\pi_1$.
        Thus, we can reduce the number of repetitions of $\pi_1$ in the initial answer path because one target node necessarily gets repeated. 
        We can moreover assume w.l.o.g. that all answer paths to $\pi$ are inner disjoints, by taking disjoint copies of the initial answer paths.
    \end{proof}
\end{toappendix}

We prove the upper-bound of Theorem~\ref{th:gpc-sat} by inductively constructing an equality type over all variables in the query.
This non-deterministic procedure uses only a polynomial amount of space by avoiding storing the full match of the pattern.
Unfortunately, this does not extend to queries using the \textsf{shortest} restrictor: they seem to require storing the full match.
We leave open the question of pinpointing the exact complexity of satisfiability for such queries. 

Given the high complexity lower bounds, one might wonder whether there are useful subclasses of GPC with tractable satisfiability. 
In Lemma~\ref{lemma:GPC-sat-single} below, we show that even under strong limitations, satisfiability is still intractable.

\begin{lemmarep}
    \label{lemma:GPC-sat-single}
    The satisfiability problem is \textsc{NP}-hard even for single-node GPC patterns.
\end{lemmarep}

\begin{proof}
    We reduce 3-SAT to the satisfiability of a GPC query. Let $F$ be the following 3-SAT formula over $n$ variables and $m$ clauses:
    $$\bigwedge\limits_{1 \leq i \leq m}\left( l_{i1} \vee l_{i2} \vee l_{i3}\right)$$ 
    where for all $1\leq i\leq m$ and $j \in \{1,2,3\}$, $l_{ij}$ is a \emph{literal} which is either $x_k$ or $\bar{x}_k$ for a $k \in \{1, \dots, n\}$.

    We construct the following GPC query $Q(x) \coloneqq \rho \: \pi_{\langle \theta \rangle}$ with  $\pi\coloneqq (x : \ell)$ and
    \begin{align*} 
        \theta  \coloneqq &\bigwedge\limits_{1 \leq i \leq n} \left( x_i = 1 \wedge \bar{x}_i = 0 \right) \vee ( x_i = 0 \wedge \bar{x}_i = 1 ) \\
                \wedge &\bigwedge\limits_{1 \leq j \leq m} \bigvee\limits_{ (b_1, b_2, b_3) \in C^3 } \left( l_{j1} = b_1 \wedge l_{j2} = b_2 \wedge l_{j3} = b_3 \right)
    \end{align*}
    where $C^3 = \{ (0,0, 1), \dots, (1, 1, 1) \}$; $x$ and $\ell$ are optional in $\pi$; 
    the literals of $F$ are used as property  names in $Q$; and size of $Q$ is clearly polynomial in the size of $F$.

    We now show that $F$ is satisfiable if and only if there exists a property graph $G$ on which $Q(x)$ returns at least a node.
    
    \noindent ($\Rightarrow$) Assume that $F$ is satisfied by an assignment $\nu$ to $\bar{x}$. We construct a property graph containing a unique $\ell$-labeled node with identifier $o$, having the following record:
    $$\forall i \in \{1, \dots, n\} \left\{\begin{array}{lr}
        \begin{array}{l@{}} x_i \mapsto 1, \\ \bar{x}_i \mapsto 0 \end{array}         & \text{if } \nu_i = \top \\
        \begin{array}{l@{}} x_i \mapsto 0, \\ \bar{x}_i \mapsto 1 \end{array}         & \text{if } \nu_i = \bot 
        \end{array}\right\}. $$ 
    By design, the top-most conjunct of $\theta$ is satisfied. Let $C_i$ for any $1 \leq i \leq m$ be a clause in $F$; by hypothesis $c_i$ is satisfied by $\nu$, so there exists $(b_1, b_2, b_3) \in C^3$ such that $\delta(o, l_{ij}) = b_j$ for all $1\leq j \leq 3$. 
    
    \noindent ($\Leftarrow$) Conversely, if $Q(x)$ is satisfied in a property graph $G$ for an element $o$; we have that $o \in N$ and each $x_i, \bar{x}_i$ are defined in the record of $o$.
    The restrictions enforced by the top-most conjunct of $\theta$ ensure that we can retrieve a well-defined assignment for $F$; the last conjunct ensures that this is a valid assignment for $F$.
\end{proof}

\subsection{Back to consistency}
\label{subsec:backtocons}

From Theorem~\ref{th:gpc-sat},  Lemma~\ref{lemma:reduc}, and Lemma~\ref{lemma:reduc-converse}, we obtain the following fundamental result.

\begin{corollary}
    \label{corol:main-result}
    The consistency problem is \textsc{PSpace}-complete for transformations using only $\mathsf{simple}$ and $\mathsf{trail}$ restrictors. 
\end{corollary}

In fact, the \textsc{PSpace} lower bound holds already for transformations using only two rules and any single restrictor. From Lemma~\ref{lemma:GPC-sat-single} and Lemma~\ref{lemma:reduc} it  follows that the problem remains intractable even for transformations using very restricted GPC queries.

In the light of these high complexity lower bounds, it is unlikely that conflict detection can be handled statically in practice. This means that conflicts have to be handled dynamically, when the transformation is executed. In Section~\ref{sec:translation} we discuss how this can be implemented in practice and in Section~\ref{sec:experiments} we show experimentally that the incurred overhead is affordable.

\section{Translation to Cypher}
\label{sec:translation}

Algorithm~\ref{alg:exec} can be seen as an abstraction of a transformation engine: it takes a transformation and an input property graph, and produces an output property graph. In this section we show how to compile a transformation to an openCypher script that can be directly executed in any openCypher engine. This is similar in spirit to executable SQL scripts for relational schema mappings, scalable and efficient in producing target solutions~\cite{bernstein_model_2007}.

We first discuss the overall complexity of Algorithm~\ref{alg:exec}. 
Lines~\ref{line:id} and~\ref{line:Eid} involve a set-theoretic union and, without appropriate optimization, their cost is proportional to the current number of elements in $T(G)$ in each iteration of the loop.
Lines~\ref{line:labels}--\ref{line:record} and ~\ref{line:Eendpoints}--\ref{line:Erecord} can be implemented in $\mathcal{O}(1)$ provided that Lines~\ref{line:id} and~\ref{line:Eid} return a pointer to the element $o \coloneqq f(C.\mbox{Id}(\bar{o})) \in T(G)$.
Thus the overall complexity of Algorithm~\ref{alg:exec} on input $G$ is:
\begin{equation} \label{eq1} 
    \mathcal{O} \left( t_{int} + n_c \cdot Int(G,T) \cdot |T(G)| \right)
\end{equation}
where $n_c$ is the total number of content constructors in $T$, $Int(G,T)$ and $t_{int}$ are respectively the total size of all intermediate results $\llbracket P \rrbracket_{G}^{\bar{x}}$ and the overall running time for computing $\llbracket P \rrbracket_{G}^{\bar{x}}$, with $P(\bar x)$ ranging over all left-hand sides of rules in $T$.

Thus, the total time taken by Algorithm~\ref{alg:exec} implemented naively is quadratic in the size of the property graphs, which makes it practically unusable for large input instances.
However, the complexity heavily depends on the implementation of the set-theoretic unions.

\begin{figure*}[t]
    \centering\small
    \begin{subfigure}[b]{.95\textwidth}
    \begin{align*}
        \underset{\langle u.address = a.aid, \: u.address = \ell.aid \rangle}{ 
            (u : \mathsf{User}), 
            (a : \mathsf{Address}), 
            (\ell : \mathsf{Location})
        } \implies 
        \underset{\langle name = u.name \rangle}{ \left(x = (u) : \mathsf{Person}\right) } & \: \underset{}{ \arr{->}{ \mathit{(\mathsf{HasLocation})} \, : \, \mathsf{HasLocation}}{} } \:
        \underset{\langle name = \ell.countryName, \: code = \ell.countryCode \rangle}{ \left((\ell.countryName) : \mathsf{Country}\right) }, \\ 
            (x) & \: \underset{}{ \arr{->}{ \mathit{(\mathsf{HasAddress})} \, : \, \mathsf{HasAddress}}{} } \:
        \underset{\langle name = a.cityName, \: code = a.cityCode \rangle}{ \left((a.cityName) : \mathsf{City}\right) }
    \end{align*}
    \end{subfigure}
    \vspace{-0.5em}
    \caption{Refined property graph transformation $T_{\mathrm{r}}$.}
    \vspace{-0.7em}
    \label{fig:refined}
\end{figure*}

\looseness=-1
\emph{Plain implementation.} In Figure~\ref{listing:synthesized} we showcase the result of our translation strategy for the variant $T_{\mathrm{r}}$ of $T$, presented in Figure~\ref{fig:refined}.
This transformation has only one rule and is translated into a single executable script.
For transformations with several rules, each rule of the transformation is independently translated into a script.

Cypher's built-in \mintinline{cypher}|elementId| primitive provides access to the identifier of an element, which is unique among all elements in the database. 
It plays a crucial role in our implementation as we actively use these identifiers as arguments to the Skolem function generating output identifiers.
To the best of our knowledge, there is no explicit control of the creation of new identifiers in Neo4j, 
so we equip nodes and edges in the output graph with a special property \verb|_|\mintinline{cypher}|id| that plays the role of controllable element identifier.

Lines~\ref{c:debutLeft}--\ref{c:finLeft} correspond to the left part of the rule and are responsible for retrieving the necessary information from the input property graph.
Recall that, in Line~\ref{line:split} of Algorithm~\ref{alg:exec}, a node rule is added  for each endpoint of every edge constructor in the transformation.
Accordingly, in the openCypher script, each node constructor used on the right-hand side of the rule is considered separately (Lines~\ref{c:debutN}--\ref{c:finN}).
Similarly to how Skolem functions are usually implemented in relational data exchange for schema mapping tasks~\cite{bernstein_model_2007}, 
we implement them with string operations, e.g., 
\verb|_|\mintinline{cypher}|id: "(" |\verb|+|\mintinline{cypher}| elementId(u) |\verb|+|\mintinline{cypher}| ")".|
\noindent 
We rely on the semantics of Cypher's \mintinline{cypher}|MERGE| clause, described in~\cite{green_updating_2019}, to implement the set-theoretic union:
in Lines~\ref{c:debutN},~\ref{c:mergeCo}, and~\ref{c:mergeCi}, \mintinline{cypher}|MERGE| checks whether an element with this identifier already exists in the graph;
either one exists and is retrieved, or a new element is created.
Adding the corresponding label(s) to the retrieved node (Line~\ref{line:labels} of Algorithm~\ref{alg:exec}) is implemented with the native Cypher's \mintinline{cypher}|SET| clause in Lines~\ref{c:labelP},~\ref{c:labelCo}, and~\ref{c:labelCi}. 
Similarly, the properties of the nodes (Line~\ref{line:record} of Algorithm~\ref{alg:exec}) are set in Lines~\ref{c:recP},~\ref{c:recCo}, and~\ref{c:finN}. 

Finally, the relationships are created (Lines~\ref{c:debutE}--\ref{c:finE}).
To keep the value of \verb|_|\mintinline{cypher}|id| unique among all elements in the output, and given the restriction that relationships hold a single label in Neo4j, the edge labels have been provided as arguments to the Skolem functions in Figure~\ref{fig:refined}.
Note that, when we merge an edge pattern, we are sure that the endpoints already exist in the database.

\begin{figure}[!ht]
    \centering
\begin{minted}[xleftmargin=2em, linenos=true, fontsize=\footnotesize, escapeinside=!!]{cypher}
MATCH (u:User)                                      !\label{c:debutLeft}!
MATCH (a:Address) WHERE a.aid = u.address
MATCH (l:Location) WHERE l.aid = u.address          !\label{c:finLeft}!
MERGE (x:_dummy { _id: "(" + elementId(u) + ")" })  !\label{c:debutN}!
SET x:Person,                                       !\label{c:labelP}!
    x.name = u.name                                 !\label{c:recP}!
MERGE (y:_dummy { _id: "(" + l.countryName + ")" }) !\label{c:mergeCo}!
SET y:Country,                                      !\label{c:labelCo}!
    y.name = l.countryName, y.code = l.countryCode  !\label{c:recCo}!
MERGE (z:_dummy { _id: "(" + a.cityName + ")" })    !\label{c:mergeCi}! 
SET z:City,                                         !\label{c:labelCi}!
    z.name = a.cityName, z.code = a.cityCode        !\label{c:finN}!
MERGE (x)-[hl:HasLocation {                         !\label{c:debutE}!
    _id: "(" + elementId(x) + "," + "HasLocation" + ","
             + elementId(y) + ")" }]->(y)
MERGE (x)-[ha:HasAddress {                          !\label{c:mergeHA}!
    _id: "(" + elementId(x) + "," + "HasAddress" + ","
             + elementId(z) + ")" }]->(z)           !\label{c:finE}!
\end{minted}
\vspace{-0.5em}
\caption{openCypher script corresponding to $T_{\mathrm{r}}$ (Figure~\ref{fig:refined}).}
\label{listing:synthesized}
\end{figure}

We point out that the \verb|_|\mintinline{cypher}|id| property and the \verb|_|\mintinline{cypher}|dummy| label are internal data; 
they are of no interest to the end user and can be dropped after the transformation  with Cypher's \mintinline{cypher}|REMOVE|~\mbox{command}.

\emph{Optimizations.} Optimizing the \mintinline{cypher}|MERGE| clauses in Lines~\ref{c:debutN},~\ref{c:mergeCo},~\ref{c:mergeCi},~\ref{c:debutE}, and~\ref{c:mergeHA} which implement the set-theoretic unions is crucial in reducing the overall execution time of the transformation.

As is the case in most database management systems, Neo4j provides facilities for query optimization.
The two that are relevant in this context are indexes and uniqueness constraints.
An \emph{index} permits to retrieve efficiently nodes with a given label that have a specific value at a given property. 
When we know in advance that all these values are unique, we can make further use of \emph{uniqueness constraints} (UCs).
Note that in our implementation, we maintain the invariant that each \verb|_|\mintinline{cypher}|id| is unique across all elements in the output.

In the version of Neo4j Community Edition that we use for running the experiments, indexes are implemented using b-trees, which means that the cost of testing if an index with $n$ elements contains a given key is $\mathcal{O}(\log n)$.
That is, by using indexes we can improve the worst-case complexity of Algorithm~\ref{alg:exec} to:
\begin{equation} \label{eq2} 
    \mathcal{O} \left( t_{int} + n_c \cdot Int(G,T) \cdot \log |T(G)| \right)
\end{equation}
In the next section we comprehensively evaluate the advantages and disadvantages of using indexes and uniqueness constraints on nodes and relationships, defined on the label/property pair \verb|_|\mintinline{cypher}|dummy|/\verb|_|\mintinline{cypher}|id|.

\emph{Conflict detection.} The consistency problem is unfortunately \textsc{PSpace}-complete by Corollary~\ref{corol:main-result}, so we cannot efficiently
check the declarative specification at compile time.
Instead, we need to be ready for potential inconsistencies at run time. 

Figure~\ref{listing:conflict-detection} illustrates how one can detect conflicts on the property \mintinline{cypher}|code| when creating a new \mintinline{cypher}|City| node.
We use the \mintinline{cypher}|ON MATCH| subclause of the \mintinline{cypher}|MERGE| clause to perform a comparison when we set a property for an existing node.
Notice that a different rule could have led to the creation of this node and, consequently, \mintinline{cypher}|z.code| may be empty; 
in this case the operator \mintinline{cypher}|<>| returns \mintinline{cypher}|false| and the correct specification is reached.

\begin{figure}[!ht]
    \centering
    \begin{minted}[xleftmargin=2em, linenos=true, fontsize=\footnotesize, escapeinside=!!]{cypher}
MERGE (z:_dummy { _id: "(" + a.cityName + ")" })
ON CREATE SET z:City, z.code = a.cityCode
ON MATCH SET z:City, z.code = CASE WHEN z.code <> a.cityCode
             THEN "Conflict detected!" !\textcolor{CypherGreen}{\textbf{ELSE}}! a.cityCode END
\end{minted}
\vspace{-0.5em}
\caption{Detecting conflicts on the property \mintinline{cypher}|code|.}
\label{listing:conflict-detection}
\end{figure}

\section{Experiments}
\label{sec:experiments}

Our experimental study has three main objectives: (i) evaluate the benefits of using this formalism for transforming property graphs in practical use-cases over a large amount of data, 
(ii) evaluate the involved overhead of detecting potential inconsistencies at run-time, and 
(iii) compare with the native openCypher approach such as the one presented in Figure~\ref{fig:motivating-scenario} (\subref{re:cypher}).

\emph{Experimental setting.} We have implemented our property graph transformations in openCypher 9 using a local Neo4j Community Edition instance in version 5.9.0.
For monitoring the results and performing the database management tasks required in our methodology, we have used Python 3.11 and the official Neo4j Python Driver 5.9.0.
The source code, datasets, and configuration files are available on the public \textsc{GitHub} repository of the project.
We performed the experiments on an HP EliteBook 840 G3 with an Intel Core i7-6600U CPU and 32GiB of system memory (2133 MHz).

\emph{Datasets.} Due to the lack of benchmarks for property graph transformations, in order to build realistic scenarios  we have adapted the mappings from several relational data integration scenarios from the iBench  suite~\cite{arocena_ibench_2015}.
In particular, we encode relational input instances as property graphs by creating a node for each tuple (no edges), and we let the target instances be property graphs as well, thus simulating graph-to-graph transformations. 
Each mapping in a scenario corresponds to a rule of our formalism. 
Following the method described in Section~\ref{sec:translation}, we compute an openCypher script implementing each rule.

The middle part of Table~\ref{tab:scenarios} reports the number $|\mathcal{L}_{in}|$ of input labels in each scenario (corresponding to the number of different relations in the original iBench scenario), the number $|\mathcal{L}^{node}_{out}|$ of output node labels, the number $|\mathcal{L}^{edge}_{out}|$ of output edge labels, and the number $|\mathcal{K}|$ of properties.
The right part provides information about the number of rules in the scenario $|T|$ and the total number $n_c$ of content constructors.
In each scenario, for each of the $|\mathcal{L}_{in}|$ input node labels, we generated up to $10^5$ nodes.

\begin{table}[t]
    \caption{Scenarios characteristics.}
    \vspace{-1em}
    \centering\footnotesize
    \begin{tabular}{l|cccc|cc}
        \toprule
        \multicolumn{1}{l}{}       & \multicolumn{4}{c}{Labels / Properties}                       & \multicolumn{2}{c}{Rules}  \\
        Scenario                   & $|\mathcal{L}_{in}|$ & $|\mathcal{L}^{node}_{out}|$ 
                                   & $|\mathcal{L}^{edge}_{out}|$ & $|\mathcal{K}|$                & $|T|$ & $n_c$              \\
        \midrule
        PersonAddress              & 2 & 2 & 1 & 7                                                 & 2 & 6                      \\
        FlightHotel                & 2 & 3 & 2 & 5                                                 & 1 & 7                      \\
        PersonData                 & 3 & 3 & 2 & 3                                                 & 1 & 5                      \\
        GUSToBIOSQL                & 7 & 5 & 4 & 80                                                & 8 & 18                     \\
        DBLPToAmalgam1             & 7 & 5 & 4 & 140                                               & 10 & 22                    \\
        Amalgam1ToAmalgam3         & 15 & 2 & 1 & 128                                              & 8 & 22                     \\
    \end{tabular}
    \vspace{-0.5em}
    \label{tab:scenarios}
\end{table}

\begin{figure}[!ht]
    \includegraphics[width=8.5cm]{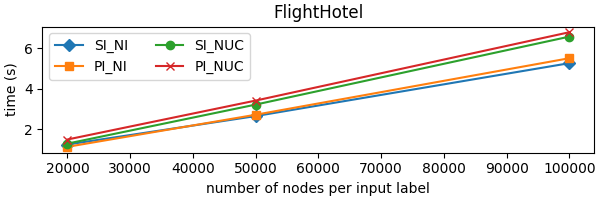}
    \caption{Comparison between uniqueness constraints and indexes for computing $T(G)$.}
    \label{fig:results-uc}
\end{figure}

\begin{figure}[!ht]
    \includegraphics[width=8cm]{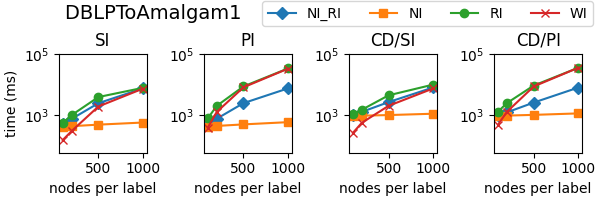}
    \vspace{-0.5em}
    \caption{Impact of indexing strategies and implementation variants on the computation of $T(G)$.}
    \label{fig:results-index}
\end{figure}

\begin{figure}[!ht]
     \centering
    \includegraphics[width=6cm]{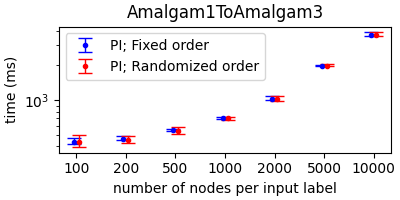}
    \vspace{-0.5em}
    \caption{Average computation time for different orders of execution of the rules.}
    \label{fig:results-random}
\end{figure}

\begin{figure}[!ht]
    \includegraphics[width=8cm]{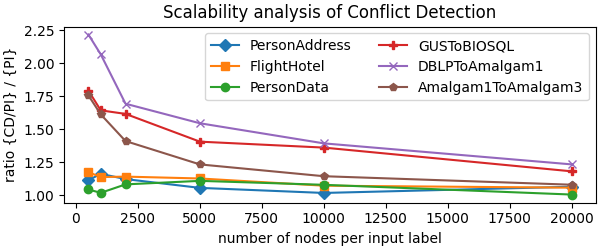}
    \vspace{-0.5em}
    \caption{Ratio between the time for computing $T(G)$ with and without conflict detection  (using PI\_NI).}
    \label{fig:results-CD}
    \vspace{-0.5em}
\end{figure}

\emph{Methodology.} The main abstraction in our implementation is a \emph{Scenario} which describes an input property graph database that contains some data of a given size stored in specific node and relationship properties.
As previously shown in Figure~\ref{fig:motivating-scenario} (\subref{re:input}), given the iBench output, we create a node for each tuple, having as properties (key/value pairs) the columns names and column values. 
We also add the Cypher specification of a set of indexes and constraints on the output side, that are created before executing the transformation when the output data is still empty. This step is not time consuming and takes on average less than one millisecond per index.

\looseness=-1
A scenario includes several Cypher queries---one for each transformation rule---that are successively applied.
To simulate the process of transforming one graph into another, and to distinguish between input and output data, we have used disjoint sets of labels in the input and output instance. 
Thus, a single database instance holds both input and output data at a time, but contains initially no output data.
As a final step, a scenario is responsible for flushing the database and removing the indexes and constraints in order to have a fresh database instance before executing the next scenario.
Note that the query cache (execution plans) is cleared when one of them is dropped.
We monitored the total amount of time spend by Neo4j in applying the transformation rules.
Each experiment generally represents the average taken over $5$ runs of a scenario.

\emph{Alternative implementation using separate indexes.} 
In Section~\ref{sec:translation}, we discussed an implementation of the framework, the \emph{Plain implementation} (PI), which uses a single index on the output side to speed up the retrieval of already existing nodes by Cypher's \mintinline{cypher}|MERGE| clause.
Using a single index for all nodes in the output may severely impact the performance of the implementation as the cost of index maintenance may become prohibitive.
To quantify this, we compare with an alternative implementation, the \emph{Separate indexes implementation} (SI), where the label is part of the argument list, similar to the case of relationships.
The goal here is to mitigate the cost of maintaining a very large index by splitting the data into many smaller ones. 
Note that it is still possible to detect conflicts in this variant with a slight modification of the code from Figure~\ref{listing:conflict-detection}.

\emph{Impact of indexes and uniqueness constraints.} 
We start by comparing the advantages of using uniqueness constraints on nodes (NUC) and indexes (NI) on the two alternative implementations, NI and SI.
Figure~\ref{fig:results-uc} reports the results for our $\mathsf{FlightHotel}$ scenario, showing that
for large input data, indexes tend to outperform UCs.

We next investigate the impact of using combinations of indexes on nodes and relationships.
We compared variants with \emph{indexes on nodes and relationships} (NI\_RI), \emph{indexes on nodes only} (NI), \emph{indexes on relationships only} (RI), and \emph{without indexes} (WI)
for the previous PI and SI implementations and their respective variants with conflict detection enabled: \emph{Conflict Detection over Plain implementation} (CD/PI), \emph{Conflict Detection over Separate indexes} (CD/SI).
We showcase in Figure~\ref{fig:results-index}, on a logarithmic scale, the results that were obtained for the 
$\mathsf{DBLPToAmalgam1}$ scenario. 
Other scenarios show similar trends and they are reported in the appendix~\cite{TPG-github}.
It is clear from the figure that the choice of indexes to use is crucial.
Using indexes only on nodes is more efficient than using a combination of indexes on nodes and relationships, which is in turn more efficient than using indexes only on relationships or using no index at all.
The key reason of this behavior is that indexes on nodes already allow accessing the endpoints of edges, along with the edges themselves, efficiently. Additional indexes on edges do not help, but do incur additional overhead.

The positive point that emerges from this study is that the implementation does not require fine tuning to be efficient in a specific scenario; using indexes only on nodes is consistently the best approach to use.
Additionally, when using indexes only on nodes (NI), the Plain implementation (PI) is negligibly slower than Separate indexes implementation (SI), whereas for other combinations of indexes it is noticeably slower.
We discussed in Example~\ref{ex:running-example} that PI allows for more flexible use of labels compared to the SI (which corresponds to having a dedicated Skolem function for each set of labels).
In view of the above results, in the remaining experiments we focus on the Plain implementation with node indexes (PI\_NI).

\looseness=-1
\emph{Impact of rule order.} Our formalism is declarative and does not specify the order for the execution of the rules.
Hence, we have investigated the impact of different orders on the computation time of the transformation.
We compare the minimum, average and maximum running times using random orders with the (fixed) order provided in iBench as baseline.
Figure~\ref{fig:results-random} reports the results for the $\mathsf{DBLPToAmalgam1}$ scenario; error bars indicate minimum and
maximum computation times observed over $20$ independent runs. 
For space reasons, $\mathsf{GUSToBIOSQL}$ and $\mathsf{Amalgam1ToAmalgam3}$, exhibiting similar results, are deferred to the appendix~\cite{TPG-github}.

We can observe that the impact of the order in which the rules are applied on the execution time of the transformation is not substantial; randomized orders have a variance similar to that of a fixed order.
It is fair to say that the performance of our implementation does not rely on any specific execution order.

\emph{Overhead of detecting potential inconsistencies.} We evaluated the impact of turning on conflict detection (over PI\_NI) by investigating the ratio between computation time with and without conflict detection.
The theoretical complexity of our implementation of Algorithm~\ref{alg:exec} with conflict detection is:
\begin{equation} \label{eq3} 
    \mathcal{O} \left( t_{int} + n_c \cdot Int(G,T) \cdot \left( \log |T(G)| + c \right) \right)
\end{equation}
for $c$ a constant modeling the cost of the conditional statement. Thus the overhead incurred by detecting conflicts is $1 + \frac{c}{\log |T(G)|}$, which tends to $1$ in larger scenarios.

The results presented in Figure~\ref{fig:results-CD} experimentally validate that the incurred overhead of conflict detection is reasonably low for large input instances, and stays within a constant factor, roughly between $1$ and $1.3$, depending on the scenario.

\begin{figure}[!ht]
    \includegraphics[width=8.5cm]{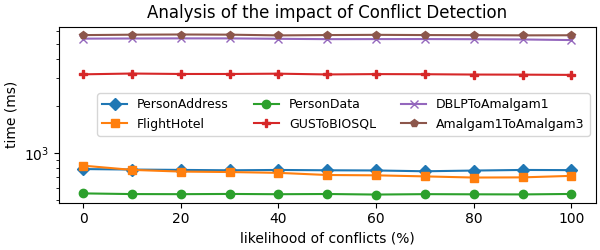}
    \vspace{-1.5em}
    \caption{Run-time comparison for different likelihood of conflicts (using CD/PI\_NI with $10^5$ nodes for each input type).}
    \vspace{-0.5em}
    \label{fig:results-ER}
\end{figure}

\emph{Robustness against incidence of conflicts.} iBench's scenarios have very few or no conflicts.
To investigate the generalizability of these results to more conflict-prone scenarios, we designed an experiment using an additional randomization step: when a rule attempts to set a value for an attribute, the value is changed randomly. 
This allows us to control the average number of conflicts in the output.
Figure~\ref{fig:results-ER} reports, on a logarithmic scale, the results for all our scenarios, with varying likelihood of conflicts, ranging from $0\%$ to $100\%$.
Note that, the size of the output is preserved because only the attributes are affected, not the topology of the graph.

We observe that the prevalence of conflicts has no impact on the execution time, suggesting that our framework's stability is preserved, even with a large proportion of conflicts in the output.

\begin{figure}[!ht]
    \vspace{-0.5em}
    \begin{subfigure}[t]{4.2cm}
        \centering
        \includegraphics[width=4.2cm]{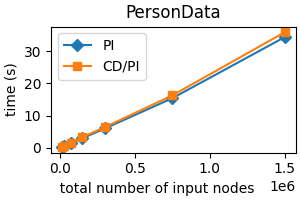}
    \end{subfigure}
    \hfill
    \begin{subfigure}[t]{4.2cm}
        \centering
        \includegraphics[width=4.2cm]{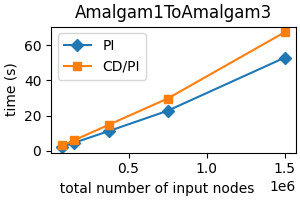}
    \end{subfigure}
    \caption{Horizontal scaling, with varying number of independent copies of the scenario.}
    \vspace{-1em}
    \label{fig:results-HS}
\end{figure}

\looseness=-1 \emph{Horizontal scalability.} We have investigated how well our framework scales with the number of rules and input labels.
We built larger scenarios by taking an increasing number of independent copies of the scenarios from Table~\ref{tab:scenarios}.
The resulting transformations reach over one hundred  rules and input labels, and over 1.5 million input nodes (in total).
Figure~\ref{fig:results-HS} reports the results for the $\mathsf{PersonData}$ and $\mathsf{Amalgam1ToAmalgam3}$ scenarios.

\looseness=-1 We observe the running time scales smoothly (almost linearly) as the number of copies increases. 
Results on other scenarios follow similar trends and are deferred to the appendix~\cite{TPG-github}.

\vspace{0.1cm}
\looseness=-1 \emph{Improvement over handcrafted scripts: a user study.} To compare empirically the readability and usability of the script-based approach and our framework, we ran an ad-hoc user study involving $12$ participants that were all already familiar with openCypher.

\looseness=-1 We compared the ability of the participants to understand the behavior of some provided openCypher scripts and transformations in clearly defined scenarios.
Only 25\% of the participants have been able to fully understand the behaviour of the openCypher scripts, whereas 67\% of them succeeded with transformations. 
In average, participants have scored 50\% on openCypher scripts and 90\% on our framework. 
Participants were also asked to compare openCypher scripts and our framework in terms of understandability, intuitiveness, and flexibility; they all have favored our framework by a great margin.
For space reasons, the questionnaire, the participant's answers, and the full discussion of the obtained results are deferred to the appendix~\cite{TPG-github}.

\begin{toappendix}
    \section{Experiments}
    \label{apx:experiments}

    \emph{User study.} The User study consists of four parts, respectively aiming to:
    \begin{itemize}
        \item Evaluate the Understandability of openCypher scripts; we asked four yes/no questions asking to the participants whether an assertion is true or not w.r.t. the behavior of a provided script in a concrete transformation scenario (e.g. Does this script create as many Director nodes as there are Person nodes that have an outgoing relationship of type DIRECTED to a Movie node?).
        \item Evaluate the Understandability of Property Graph Transformations; with four similar yes/no questions on a slightly different transformation to avoid biases;
        \item The third part required participants to modify some provided openCypher scripts and transformations to adapt to a new requirement. We collected the participants’ answers and checked them.
        \item The last part required the participants to give their opinion on the following questions and to indicate, in a range for 1 to 5 (3 is neutral) whether they found openCypher scripts and/or transformation rules better on:
        \begin{itemize}
            \item Which one of the two methods do you find easier to understand?
            \item Which one of the two methods do you find more intuitive? (Better for describing the desired output.)
            \item Which one of the two methods do you find more flexible? (Easier to adapt to a new specification.)
        \end{itemize}
    \end{itemize}

    We collected the answers of 12 participants, that were asked to self report their level of expertise in a range from (1 - Novice) to (5 - Expert) on the following topics:
    \begin{itemize}
        \item How would you rate your knowledge about databases? \newline
        The answers filled in by the participants ranged from 3 to 5, included.
        \item How would you rate your knowledge about openCypher? \newline
        The answers filled in by the participants ranged from 2 to 4, included.
        \item How would you rate your knowledge about the MERGE clause of openCypher? \newline
        The answers filled in by the participants ranged from 1 to 5, included.
        \item How would you rate your knowledge about property graph transformations? \newline
        The answers filled in by the participants ranged from 1 to 3, included.
    \end{itemize}
    We have a pool of people that all have prior exposure to openCypher (2-4) but a great diversity w.r.t. the knowledge of the MERGE clause of openCypher (the basic tool for updates in openCypher), i.e. from novice to expert.

    The results on the first two parts are as follow:
    \begin{itemize}
        \item The average number of correct answers is 50\% (2.0 out of 4) for the understandability of the openCypher scripts, and 90\% (3.6 out of 4) for the understandability of the transformation rules.
        \item 25.0\%, resp. 67\% of participants checked all the correct answers in the first, resp. second part.
        \item All participants scored higher on their individual understanding of transformation rules compared to openCypher scripts.
    \end{itemize}
    Given those results, it is extremely clear that it is very difficult for people to understand even the basic openCypher scripts used to transform property graphs, whereas our framework – despite being absolutely new to the participants, has been widely understood.

    The results on the last part are as follow (recall that 3 is neutral, 1 is strong preference for openCypher scripts and 5 is strong preference for transformation rules):
    \begin{itemize}
        \item Which one of the two methods do you find easier to understand? \newline
        Collected answers range from 1 to 5 with an average of 3.3.
        \item Which one of the two methods do you find more intuitive? (Better for describing the desired output.) \newline
        Collected answers range from 3 to 5 with an average of 3.8.
        \item Which one of the two methods do you find more flexible? (Easier to adapt to a new specification.) \newline
        Collected answers range from 3 to 5 with an average of 4.1.
    \end{itemize}

    Moreover, we have noticed that the participants having a low understanding of openCypher scripts (scored 2 or less out of 4 in the first part) have been more inclined to provide less credit to the transformation rules than other people. 
    So we decided to split the participants in two groups to investigate this more.

    The results on the last part are as follow only for the 8 people that have score 2 or lower out of 4 in their understanding of openCypher scripts (recall that 3 is neutral, 1 is strong preference for openCypher scripts and 5 is strong preference for transformation rules):
    \begin{itemize}
        \item Which one of the two methods do you find easier to understand? \newline
        Collected answers range from 1 to 5 with an average of 2.75.
        \item Which one of the two methods do you find more intuitive? (Better for describing the desired output.) \newline
        Collected answers range from 3 to 5 with an average of 3.9.
        \item Which one of the two methods do you find more flexible? (Easier to adapt to a new specification.) \newline
        Collected answers range from 3 to 5 with an average of 3.9.
    \end{itemize}

    The results on the last part are as follow only for the 4 people that have score 3 or higher out of 4 in their understanding of openCypher scripts (recall that 3 is neutral, 1 is strong preference for openCypher scripts and 5 is strong preference for transformation rules):
    \begin{itemize}
        \item Which one of the two methods do you find easier to understand? \newline
        Collected answers range from 1 to 5 with an average of 4.7.
        \item Which one of the two methods do you find more intuitive? (Better for describing the desired output.) \newline
        Collected answers range from 3 to 5 with an average of 3.7.
        \item Which one of the two methods do you find more flexible? (Easier to adapt to a new specification.) \newline
        Collected answers range from 3 to 5 with an average of 4.7.
    \end{itemize}
    It is therefore clear that those who have been less convinced that openCypher scripts are more error-prone and harder to interpret and analyze have not figured out for themselves that openCypher scripts are difficult to understand and manipulate. Moreover people with a good understanding of openCypher are clearly in favor that our framework is easier to understand.

    With this study, we empirically and experimentally demonstrated that the script-based approach can be error prone, hard to interpret and analyze (i.e., less usable) and that the improvement of usability and accuracy over handcrafted, script-based solutions have clearly been attested by a majority of the participants.

    \begin{figure*}[!ht]  
        \begin{subfigure}[b]{0.33\textwidth}
            \centering
            \includegraphics[width=5.5cm]{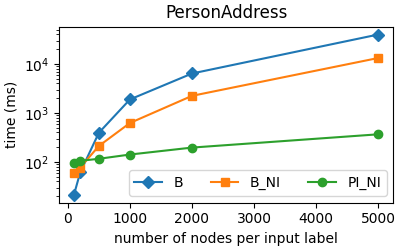}
        \end{subfigure}
        \hfill
        \begin{subfigure}[b]{0.33\textwidth}
            \centering
            \includegraphics[width=5.5cm]{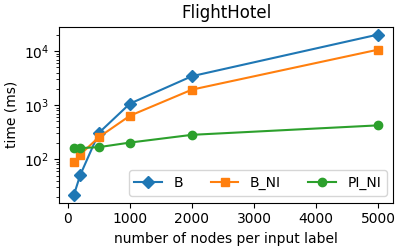}
        \end{subfigure}
        \hfill
        \begin{subfigure}[b]{0.33\textwidth}
            \centering
            \includegraphics[width=5.5cm]{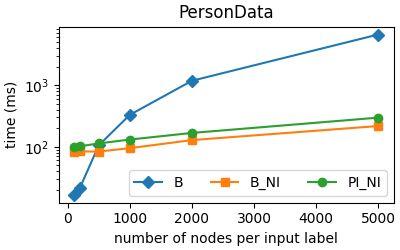}
        \end{subfigure}
        \caption{Run-time comparison with the baseline approach, depending on the number of input nodes of each type in $G$.}
        \label{fig:results-baseline}
    \end{figure*}

    \emph{Comparison with native Cypher approach.} Finally, we compared our framework (using PI\_NI) with ad-hoc transformation scripts (B-NI, B; respectively with and without node indexes), 
    such as the one presented in Figure~\ref{fig:motivating-scenario} (\subref{re:cypher}).
    The result over $\mathsf{PersonAddress}$, $\mathsf{FlightHotel}$ and $\mathsf{PersonData}$ are presented in Figure~\ref{fig:results-baseline}.
    For larger scenarios, such as $\mathsf{GUSToBIOSQL}$, $\mathsf{DBLPToAmalgam1}$ and $\mathsf{Amal}\-\mathsf{gam1ToAmalgam3}$, handcrafted transformation scripts are exceedingly large due to the number of rules and properties involved.

    We can observe that our solution clearly outperforms the handcrafted solutions in most of the cases. 
    The only exception occurs when using the $\mathsf{PersonData}$ scenario, for which the B-NI baseline is slightly better than our solution, 
    while the B baseline is still outperformed. The underlying reason is due to the nature of this scenario, for which the \mintinline{cypher}|collect| clause contains only one element.

    \begin{figure}[!ht]
        \begin{subfigure}[b]{.45\textwidth}
            \centering
            \includegraphics[width=8cm]{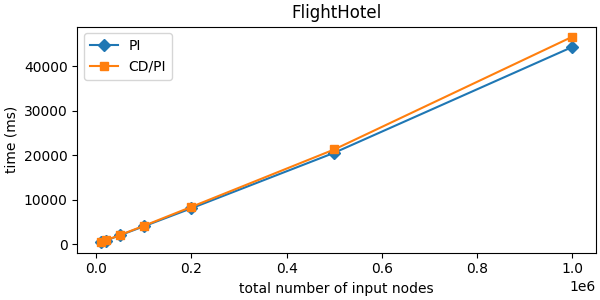}
        \end{subfigure}
        \hfill
        \begin{subfigure}[b]{.45\textwidth}
            \centering
            \includegraphics[width=8cm]{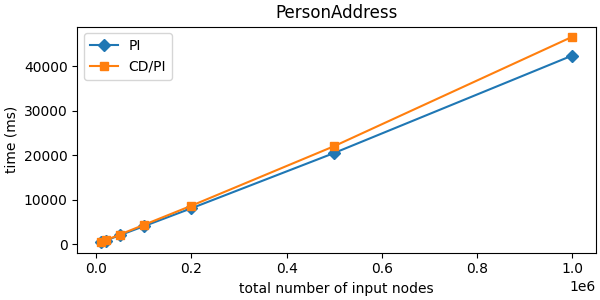}
        \end{subfigure}
        \hfill
        \begin{subfigure}[b]{.45\textwidth}
            \centering
            \includegraphics[width=8cm]{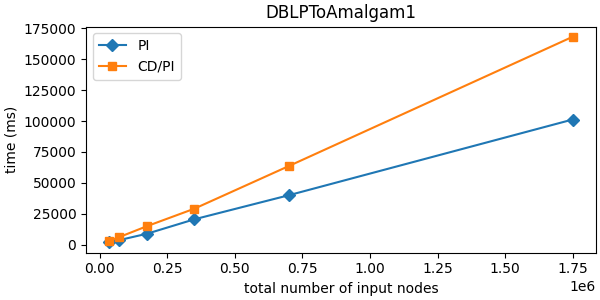}
        \end{subfigure}
        \hfill
        \begin{subfigure}[b]{.45\textwidth}
            \centering
            \includegraphics[width=8cm]{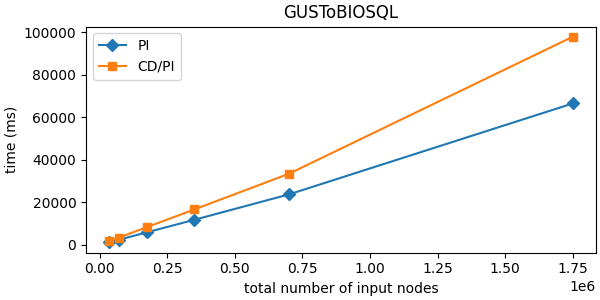}
        \end{subfigure}
        \caption{Horizontal scaling, with varying number of independent copies of the scenario.}
        \label{fig:results-HS-apx}
    \end{figure}

    \begin{figure*}[!ht]
        \begin{subfigure}[b]{0.5\textwidth}
            \captionsetup{justification=centering}
            \centering
            \includegraphics[width=7cm]{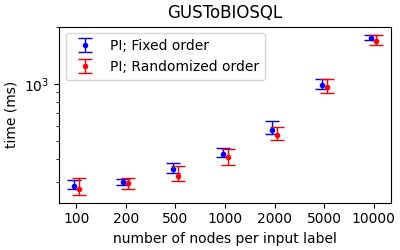}
        \end{subfigure}
        \hfill
        \begin{subfigure}[b]{0.5\textwidth}
            \captionsetup{justification=centering}
            \centering
            \includegraphics[width=7cm]{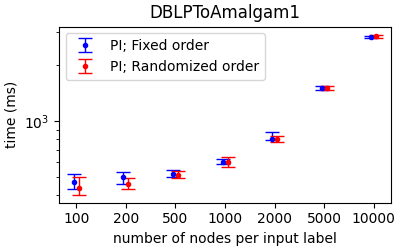}
        \end{subfigure}
        \caption{Average computation time for different orders of execution of the rules; error bars indicate minimum and maximum computation times observed over $20$ independent runs.}
        \label{fig:results-random-apx}
    \end{figure*}

    \begin{figure*}[!ht]
        \begin{subfigure}[b]{1\textwidth}
            \centering
            \includegraphics[width=17cm]{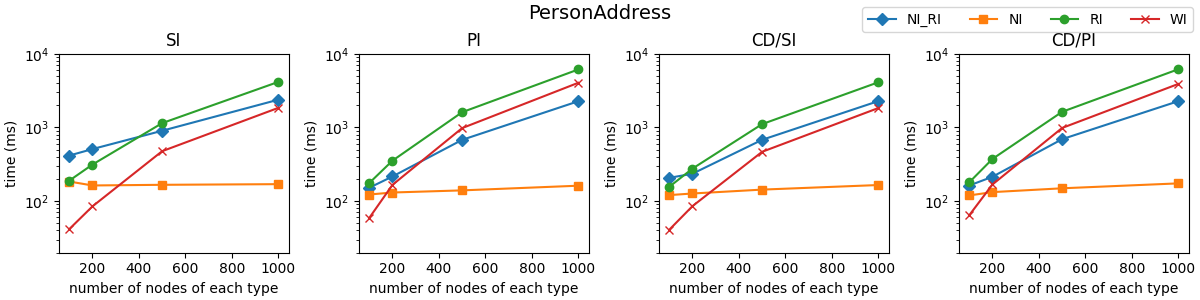}
        \end{subfigure}
        \hfill
        \begin{subfigure}[b]{1\textwidth}
            \centering
            \includegraphics[width=17cm]{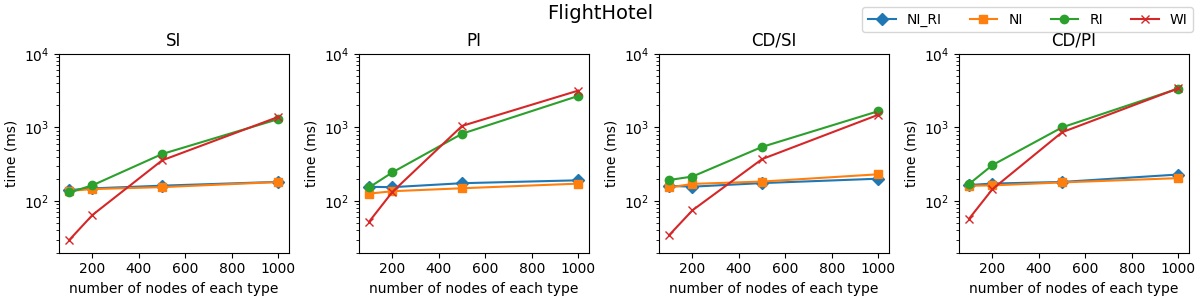}
        \end{subfigure}
        \hfill
        \begin{subfigure}[b]{1\textwidth}
            \centering
            \includegraphics[width=17cm]{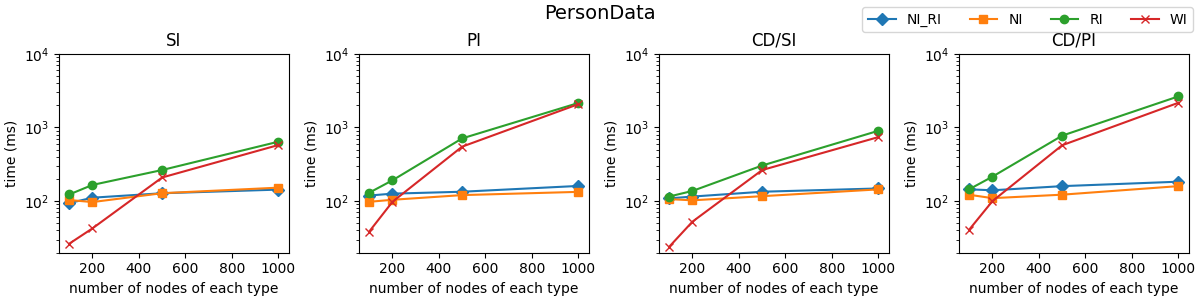}
        \end{subfigure}
        \hfill
        \begin{subfigure}[b]{1\textwidth}
            \centering
            \includegraphics[width=17cm]{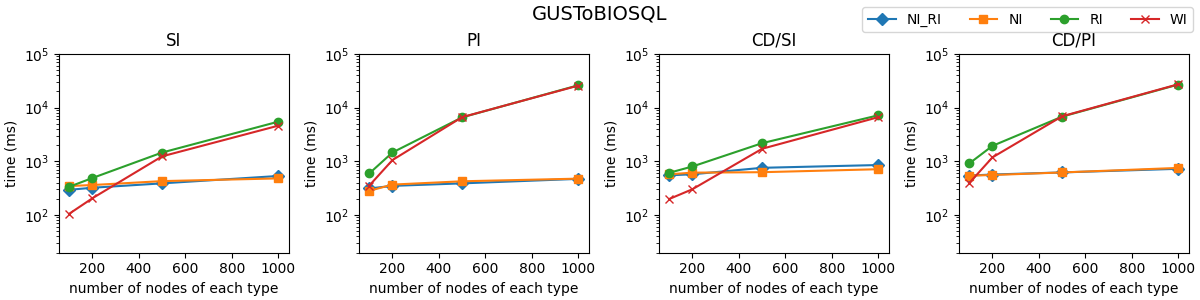}
        \end{subfigure}
        \hfill
        \begin{subfigure}[b]{1\textwidth}
            \centering
            \includegraphics[width=17cm]{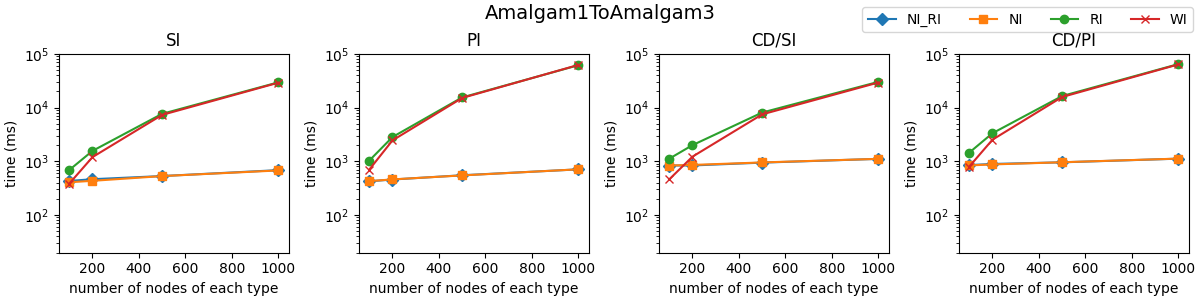}
        \end{subfigure}
        \caption{Impact of indexing strategies and implementation variants on the computation of $T(G)$.}
        \label{fig:results-index-apx}
    \end{figure*}
\end{toappendix}

\begin{table}[t]
    \caption{Running times and size of intermediate data (ICIJ).}
    \vspace{-1em}
    \centering\footnotesize
    \begin{tabular}{l|rr|rrr|rr}
        \toprule
        Rules                       & $t_{int}$ & $t$ & $Int(G,T)$ & $|T(G)|$ & $O/I$                     & $t_{int}^e$ & $t^e$              \\
        \midrule
        $R1-R4$                     & 2,757 & 11,192 & 374,955 & 748,524                                & 1.996 & 0.007 & 0.015            \\
        $R5-R9$                     & 3,553 & 5,946 & 62,242 & 82,616                                   & 1.327 & 0.057 & 0.072            \\
        $R10-R13$                   & 15,509 & 36,775 & 1,906,686 & 1,905,547                           & 0.999 & 0.008 & 0.019            \\
        $R14-R17$                   & 9,667 & 21,006 & 493,556 & 1,173,720                              & 2.378 & 0.020 & 0.018            \\
        $R18$ only                  & 8,407 & 25,640 & 785,124 & 1,570,470                              & 2.000 & 0.011 & 0.016            \\
    \end{tabular}
    \vspace{-2em}
    \label{tab:icij} 
\end{table}

\subsection{Use-case Study: Improving Data Integration}
\label{subsec:icij}

\begin{figure*}[!ht]
    \centering\small
    \begin{subfigure}[b]{.95\textwidth}
    \begin{align*}
        \underset{\langle toLower(o.name) = toLower(p.name) \rangle}{ 
            (o : \mathsf{Officer}) \underset{}{ \arr{->}{ :\mathsf{similar}}{} }^{0..\infty}
            (\: : \mathsf{Officer}) \underset{}{ \arr{->}{ : \mathsf{registered\_address}}{} }
            (\: : \mathsf{Address}) \underset{}{ \arr{->}{ : \mathsf{similar}}{} }
            (\: : \mathsf{Address}) \underset{}{ \arr{<-}{ : \mathsf{registered\_address}}{} }
            (\: : \mathsf{Officer})  \underset{}{ \arr{<-}{ :\mathsf{similar}}{} }^{0..\infty} 
            (p : \mathsf{Officer})
        } \\ 
        \implies 
        \underset{}{ \left((o) : \mathsf{T\_Officer}\right) }  \: \underset{}{ \arr{->}{ \mathit{(\mathsf{})} \, : \, \mathsf{T\_Similar}}{\langle link = \mathsf{"similar\ name\ and\ address\ as"} \rangle} } \:
        \underset{}{ \left((p) : \mathsf{T\_Officer}\right) }
    \end{align*}
    \end{subfigure}
    \vspace{-0.25cm}
    \caption{Improved similarity detection ($R15$).}
    \vspace{-0.1cm}
    \label{fig:similarity-detection}
\end{figure*}

\begin{toappendix}
    \subsection{Offshore Leaks Dataset}
    \label{apx:subsec:icij}
    \begin{figure*}[!ht]
        \centering
        \begin{subfigure}[b]{1\textwidth}
            \centering
            \captionsetup{justification=centering}
            \begin{flalign}
                \tag{$R_1$}
                \underset{\langle sourceID = \mathsf{"Malta\ registry"}, \: a.address = a.address \rangle}{ 
                    (a : \mathsf{Address})
                } \implies 
                \underset{\langle source = a.sourceID \rangle}{ \left(x = (a) : \mathsf{T\_Address}\right) } & \: \underset{}{ \arr{->}{ \, : \, \mathsf{T\_LOCATED}}{}  } \:
                \underset{\langle name = a.country \rangle}{ \left(y = (a.country\_code) : \mathsf{T\_Country}\right) } \label{rule:1} \\ 
                \tag{$R_2$}
                \underset{\langle sourceID = \mathsf{"Malta\ registry"}, \: \neg \left( a.address = a.address \right) \rangle}{ 
                    (a : \mathsf{Address})
                } \implies
                 \underset{\langle source = a.sourceID \rangle}{ \left(x = (a) : \mathsf{T\_Address}\right) } & \label{rule:2}
                \\ 
                \tag{$R_3$}
                \underset{\langle \neg \left( sourceID = \mathsf{"Malta\ registry"} \right), \: a.address = a.address \rangle}{ 
                    (a : \mathsf{Address})
                } \implies 
                \underset{\langle source = a.sourceID \rangle}{ \left(x = (a) : \mathsf{T\_Address}\right) } & \: \underset{}{ \arr{->}{ \, : \, \mathsf{T\_LOCATED}}{}  } \:
                \underset{\langle name = a.country \rangle}{ \left(y = (a.country\_codes) : \mathsf{T\_Country}\right) } \label{rule:3} \\ 
                \tag{$R_4$}
                \underset{\langle \neg \left( sourceID = \mathsf{"Malta\ registry"} \right), \: \neg \left( a.address = a.address \right) \rangle}{ 
                    (a : \mathsf{Address})
                } \implies
                \underset{\langle source = a.sourceID \rangle}{ \left(x = (a) : \mathsf{T\_Address}\right) } & \label{rule:4}
            \end{flalign}
            \caption{Refactoring registered addresses ($R1-R4$).}
            \label{sfig:1-4}
        \end{subfigure}
        \begin{subfigure}[b]{1\textwidth}
            \centering
            \captionsetup{justification=centering}
            \begin{flalign}
                \tag{$R_5$}
                (i : \mathsf{Intermediary}) \: \underset{}{ \arr{->}{ \, : \, \mathsf{registered\_address}}{}  } \: (a : \mathsf{Address})
                \implies &
                (x = (i) : \mathsf{T\_Intermediary}) \: \underset{}{ \arr{->}{ \, : \, \mathsf{T\_REG\_ADDRESS}}{}  } \: (y = (a) : \mathsf{T\_Address}) \label{rule:5} \\ 
                \notag
                \underset{\langle \neg \left(address = \mathsf{""} \right), \: address = address, \: country\_code = country\_code \rangle}{ 
                    (i : \mathsf{Intermediary})
                } \implies &
                (x = (i) : \mathsf{T\_Intermediary}) \: \underset{}{ \arr{->}{ \, : \, \mathsf{T\_REG\_ADDRESS}}{}  } \: 
                \underset{\langle source = i.sourceID \rangle}{ (y = (i.address) : \mathsf{T\_Address}) }, \\
                & \tag{$R_6$}
                \left( y \right) \: 
                \underset{}{ \arr{->}{ \, : \, \mathsf{T\_LOCATED}}{} } \:
                \underset{\langle name = i.countries \rangle}{ \left(z = (i.country\_codes) : \mathsf{T\_Country} \right) } \label{rule:6}
                \\
                \tag{$R_7$}
                \underset{\langle \neg \left(address = \mathsf{""} \right), \: address = address, \: \neg \left( country\_code = country\_code \right) \rangle}{ 
                    (i : \mathsf{Intermediary})
                } \implies &
                (x = (i) : \mathsf{T\_Intermediary}) \: \underset{}{ \arr{->}{ \, : \, \mathsf{T\_REG\_ADDRESS}}{}  } \: 
                \underset{\langle source = i.sourceID \rangle}{ (y = (i.address) : \mathsf{T\_Address}) } \label{rule:7} \\ 
                \tag{$R_8$}
                \underset{\langle \neg \left(address = \mathsf{""} \right), \: address = address, \: \neg \left( country\_code = country\_code \right) \rangle}{
                    (i : \mathsf{Intermediary}) \: \underset{}{ \arr{->}{ \, : \, \mathsf{intermediary\_of}}{}  } \: (e : \mathsf{Entity})
                } \implies &
                (x = (i) : \mathsf{T\_Intermediary}) \: \underset{}{ \arr{->}{ \, : \, \mathsf{T\_REG\_ADDRESS}}{}  } \: 
                \underset{\langle source = i.sourceID \rangle}{ (y = (e.address) : \mathsf{T\_Address}) } \label{rule:8} \\
                \notag
                \underset{\langle \neg \left(address = \mathsf{""} \right), \: address = address, \: country\_code = country\_code \rangle}{
                    (i : \mathsf{Intermediary}) \: \underset{}{ \arr{->}{ \, : \, \mathsf{intermediary\_of}}{}  } \: (e : \mathsf{Entity})
                } \implies &
                (x = (i) : \mathsf{T\_Intermediary}) \: \underset{}{ \arr{->}{ \, : \, \mathsf{T\_REG\_ADDRESS}}{}  } \: 
                \underset{\langle source = i.sourceID \rangle}{ (y = (e.address) : \mathsf{T\_Address}) }, \\
                & \tag{$R_9$}
                \left( y \right) \: 
                \underset{}{ \arr{->}{ \, : \, \mathsf{T\_LOCATED}}{} } \:
                \underset{\langle name = i.countries \rangle}{ \left(z = (e.country\_codes) : \mathsf{T\_Country} \right) } \label{rule:9}
            \end{flalign}
            \caption{Uniformizing address information for intermediaries ($R5-R9$).}
            \label{sfig:5-9}
        \end{subfigure}
        \begin{subfigure}[b]{1\textwidth}
            \centering
            \captionsetup{justification=centering}
            \begin{flalign}
                \tag{$R_{10}$}
                (i : \mathsf{Intermediary})
                \implies &
                \underset{\langle name = i.name, \: status = i.status, \: valid\_until = i.valid\_until, \: source = i.sourceID \rangle}{ \left(x = (i) : \mathsf{T\_Intermediary}\right) } \label{rule:10} \\ 
                \tag{$R_{11}$}
                (a : \mathsf{Address}) 
                \implies &
                \underset{\langle address = a.address, \: orig\_addr = a.orig\_addr, \: valid_\_until = a.valid_\_until \rangle}{ \left(x = (a) : \mathsf{T\_Address}\right) } \label{rule:11}
                \\ 
                \tag{$R_{12}$}
                (e : \mathsf{Entity})
                \implies &
                \underset{\langle name = e.name, \: orig\_name = e.orig\_name, \: inact\_date = e.inact\_date, \: inc\_date = e.inc\_date, \:
                \dots, \: source = e.sourceID \rangle}{ \left(x = (e) : \mathsf{T\_Entity}\right) } \label{rule:12} \\ 
                \tag{$R_{13}$}
                (o : \mathsf{Officer})
                \implies &
                \underset{\langle name = o.name, \: status = o.status, \: source = o.sourceID \rangle}{ \left(x = (o) : \mathsf{T\_Officer}\right) } \label{rule:13}
            \end{flalign}
            \caption{Exporting the nodes ($R10-R13$).}
            \label{sfig:10-13}
        \end{subfigure}
        \caption{Rules of the ICIJ database transformation.}
        \label{fig:icij-rules}
    \end{figure*}
    \begin{figure*}[!ht]
        \centering
        \begin{subfigure}[b]{1\textwidth}
            \centering
            \captionsetup{justification=centering}
            \begin{flalign}
                \tag{$R_{14}$}
                (o : \mathsf{Officer}) \: \underset{}{ \arr{->}{ \, : \, \mathsf{similar}}{}  } \: (p : \mathsf{Officer})
                \implies &
                (x = (o) : \mathsf{T\_Officer}) \: \underset{}{ \arr{->}{ \mathit{(\mathsf{})} \, : \, \mathsf{T\_Similar}}{\langle link = \mathsf{"similar\ name\ and\ address\ as"} \rangle} } \: 
                (y = (p) : \mathsf{T\_Officer}) \label{rule:14} \\ 
                \tag{$R_{16}$}
                (a : \mathsf{Address}) \: \underset{}{ \arr{->}{ \, : \, \mathsf{same\_as}}{}  } \: (b : \mathsf{Address})
                \implies &
                (x = (a) : \mathsf{T\_Address}) \: \underset{}{ \arr{->}{ \, : \, \mathsf{T\_SAME\_AS}}{}  } \: 
                (y = (b) : \mathsf{T\_Address}) \label{rule:16} \\ 
                \tag{$R_{17}$}
                (o : \mathsf{Officer}) \: \underset{}{ \arr{->}{ \, : \, \mathsf{registered\_address}}{}  } \: (a : \mathsf{Address})
                \implies &
                (x = (o) : \mathsf{T\_Officer}) \: \underset{}{ \arr{->}{ \, : \, \mathsf{T\_REGISTERED\_ADDRESS}}{}  } \: 
                (y = (a) : \mathsf{T\_Address}) \label{rule:17}
            \end{flalign}
            \caption{Improving similarity detection ($R14-R17$).}
            \label{sfig:14-17}
        \end{subfigure}
        \begin{subfigure}[b]{1\textwidth}
            \centering
            \captionsetup{justification=centering}
            \begin{flalign}
                \notag
                \underset{\langle jurisdiction\_desc = jurisdiction\_desc \rangle}{
                    (e : \mathsf{Entity})
                } \implies &
                (x = (e) : \mathsf{T\_Entity}) \: \underset{}{ \arr{->}{ \, : \, \mathsf{T\_IN\_JURIS}}{}  } \: 
                \underset{\langle juris = e.jurisdiction\_desc \rangle}{ (y = (e.jurisdiction\_desc) : \mathsf{T\_Jurisdiction}) }, \\
                & \tag{$R_{18}$}
                \left( y \right) \: 
                \underset{}{ \arr{->}{ \, : \, \mathsf{T\_RELATED}}{} } \:
                \left(z = (e.jurisdiction) : \mathsf{T\_Country} \right) \label{rule:18}
            \end{flalign}
            \caption{Refactoring jurisdictions ($R18$).}
            \label{sfig:18}
        \end{subfigure}
        \caption{Rules of the ICIJ database transformation. (Bis)}
        \label{fig:icij-rules-bis}
    \end{figure*}
\end{toappendix}

In this section, we want to compare the cost of running the whole transformation compared to the cost of querying the source property graph to extract the bindings (intermediate data).
To this end, we use a real-world dataset, the \emph{Offshore Leaks Database and guide from the International Consortium of Investigative Journalists (ICIJ)}~\cite{ICIJ-github}, a property graph with 1,908,466 nodes and 3,193,390 edges taken from~\cite{10.14778/3611479.3611506}.
This dataset consolidates data from several leaks (\emph{Panama Papers}, \emph{Bahamas Leaks}, etc.) collected by ICIJ over a period of ten years, but still presents the consolidated data in a heterogeneous manner.
The dataset contains information about \emph{entities} (off-shore companies), \emph{officers} of those, \emph{intermediaries} (middlemen who help set up off-shore companies), and \emph{jurisdictions} (countries or territories where off-shore companies are registered).
We have designed a modular $18$-rule transformation aiming to uniformize the presentation of the information contained in the graph.
The rules are grouped into 5 subsets, each addressing a specific refactoring goal motivated below.
For space reasons, we have deferred the rules themselves to the appendix~\cite{TPG-github}.

\emph{Refactoring registered addresses ($R1-R4$).} 
The ICIJ database contains the registered \emph{addresses} of the \emph{officers}, and \emph{entities}.
These rules are responsible for creating nodes representing \emph{countries} and linking to the \emph{addresses}. 
Because the data is semi-structured and collected from multiple sources, information may be stored in attributes that have different names, or may even not be available at all.
All these cases are covered by these four rules.

\emph{Uniformizing address information for intermediaries ($R5-R9$).} 
After careful investigation, we found that the registered \emph{address} of \emph{intermediaries} can be stored in three different ways in the database: 
(i) an \emph{intermediary} can have a direct relationship with an \emph{address}, 
(ii) the \emph{address} can be stored in the properties of the node itself, and 
(iii) when neither of the two previous cases applies, it is necessary to retrieve the \emph{address} of an \emph{entity} linked to this \emph{intermediary}. 
These rules permit to consistently store address information.

\emph{Exporting the nodes ($R10-R13$).} 
These rules copy the node information from the source to the target; they are necessary to preserve all the information from the original graph.

\emph{Improving similarity detection ($R14-R17$).} 
Because the dataset consolidates multiple leaks, certain specific relationships, such as \emph{similar} and \emph{same\_as}, are used to  indicate that some \emph{officers} (resp. \emph{addresses}) are likely to represent the same real life entity.
These rules focus on exporting this data and improving the similarity detection. 
This is illustrated by Rule $15$ shown in Figure~\ref{fig:similarity-detection} which composes the relationships \emph{similar} and \emph{same\_as} to ensure that both its endpoints correspond to \emph{officers} having the same address. 
(This is because \emph{similar} encompasses address similarity.)
Then, it checks whether their names are also similar. If both conditions hold, it safely adds a similarity edge between the endpoints in the output.

\looseness=-1 \emph{Refactoring jurisdictions ($R18$).} 
The last rule is responsible for connecting the \emph{jurisdictions} with their associated \emph{countries}; this information is not explicitly stored in the initial database.

\looseness=-1 \emph{Results.} Our experimental results are reported in Table~\ref{tab:icij}.
We report the time $t_{int}$ (in \emph{ms}) the database takes to retrieve the intermediate data; the total time $t$ of running the transformation (extracting the bindings and constructing the output); the size $Int(G,T)$ of intermediate data; the size of the output $T(G)$; the ratio $O/I$ of the size of the output to the size of intermediate data.
To account for the differences in the sizes of the outputs of the respective tasks, we also report the average time $t_{int}^e$ taken to produce each binding of the intermediate result, and the average time $t^e$ taken to construct each element of the output.
We break down the reported values into groups of rules corresponding to the aforementioned integration tasks.

There are several things that we can learn from Table~\ref{tab:icij}. 
First, the overhead $t^e / t_{int}^e$ of turning the intermediate results into a proper property graph is reasonable. 
For Rules $R14-R17$, it is even comparatively more efficient to compute the output property graph (overhead is $0.9$).
The worst case is for rules $R10-R13$ exhibiting an overhead of $2.4$.
Second, the ratio $O/I$ is also reasonable, ranging from $1$ to $2.4$.
This shows that, in practical contexts, $Int(G,T)$ can be assumed to have a size comparable to $|T(G)|$.

Thus, we have demonstrated that the overhead incurred by producing a property graph rather than a set of bindings is acceptable for a realistic transformation in a real-life integration scenario.

\section{Related Work}
\label{sec:rw}

\looseness=-1 \emph{Schema mapping and data exchange.}
Specifying the relationship between two relational (or XML) schemas using a set of declarative assertions is a task known as \emph{schema mapping}~\cite{10.1145/1065167.1065176,fagin_data_2005, bellahsene2011schema}.
This relation, is usually \emph{non functional}, i.e. given an input instance $I$, several target instances satisfying the mapping constraints exist.

\looseness=-1
Schema mappings and data exchange have been studied in~\cite{10.1145/2448496.2448520, 10.1145/3034786.3056113} for graph data\-bases.
The mapping languages considered are based on classical graph database queries such as regular path que\-ries~\cite{Barcel2012RelativeEO},
limited in their expressivity by not supporting data values.
Moreover, answering queries on the target is already intractable in data complexity for RPQs~\cite{10.1145/2448496.2448520} and undecidable for data RPQs~\cite{10.1145/3034786.3056113}.
In comparison, our transformation framework provides more flexibility by including the support for data values, and any target query can be answered by simple execution on the produced property graph.

\emph{Graph transformations.} Graph database transformations defined using Datalog-like rules based on acyclic conjunctive two-way regular path queries have been investigated in~\cite{10.1145/3584372.3588654}.
They study three fundamental static analysis problems: \emph{type checking}, \emph{equivalence} of transformations under graph schemas, and \emph{schema elicitation}.
They show all these problems to be in \textsc{ExpTime}.

A key difference with our work lies in the graph database model they consider, which does not have data values. 
We have seen that dealing with data values gives rise to the consistency checking problem, which is key to understanding if a property graph transformation is well-defined.
Moreover, their query language -- a fragment of Datalog, is not practical for querying property graphs~\cite{7ad59132cb3c45e2851f565fbb703cea}.
Another difference is that they are using a single dedicated node constructor for each label.
In Section~\ref{sec:pgt} and~\ref{sec:translation}, we have seen that this approach is too rigid for dealing with multiple labels.

\looseness=-1
\emph{Object-creating functions.} The Skolem functions we use in our constructors resemble to the object creating functions that are used in the object-oriented database model~\cite{10.1145/290179.290182, 10.5555/645916.671975}.
Among transformation languages based on oid generation, StruQL~\cite{10.1145/262762.262763} specifically operates on object-oriented semi-structured instances where nodes can either be data values or contain an oid and labeled edges can connect oid nodes to oid or value nodes.
The major difference with our work is that they have \emph{multi-valued} attributes: i.e., an oid node may be connected via $a$-edges to several value nodes.
Hence, additional integrity constraints are necessary to ensure a correct modeling of property graphs in their model.
Therefore, they did not take into account the problem of consistency.

\looseness=-1 \emph{Interoperability of graph data.} Although RDF, RDF-star and the property graph data model share striking similarities, both being based on elementary graph concepts, like nodes and edges,
intricate interoperability issues arise when attempting to exchange data between them.
RDF-star notably allows for annotating RDF triples with metadata annotations, which are notoriously difficult to capture within the property graph data model as witnessed in~\cite{abuoda_transforming_2022}.

\looseness=-1 The main concern of transformation languages between graph data models is thus primarily focused on solving the well-known impedance mismatch problem~\cite{bernstein_model_2007},
which does not arise in our setting because we have property graphs for both input and output.
Our transformation language can be thus more expressive, and can be executed by the graph database management system itself.

\emph{Mining the identities of nodes across networks.} Network alignment is a technique for finding node correspondences between two or more networks. It can be used, for example, to associate nodes from different social networks with the same user~\cite{10.1145/3340531.3412168}.
Nodes are identified based on their similarities with respect to both their features (i.e., their properties) and their neighborhood.

While these methods are not part of graph transformation formalisms, they can be used to guide the  construction of graph transformations.  
For instance, in Section~\ref{subsec:icij}, the results of network alignment (the similarity edges in the Offshore Leaks Database) were leveraged to better integrate data coming from multiple leaks.

\section{Conclusion}
\label{sec:conclusion}

Our research is the first to lay the theoretical foundations for declarative property graph transformations, and facilitate practical solutions for turning such specifications into executable scripts in modern property graph query languages. 
New challenges arise from the specification of property-aware transformations, notably the task of checking if a transformation is \emph{consistent}. 
Using a proof-of-concept implementation of our formalism in openCypher, we showcase the efficiency of our approach for transforming property graphs for both real-world and synthetic datasets. 

\looseness=-1 This work paves the way for obtaining compositional semantics for graph query languages. 
As a future direction, we will investigate the model extensions needed for the above semantics, by addressing label and path variables, and aggregates.
Meanwhile, our framework can already seamlessly support the group variables of GPC because those are list of identifiers that can be flattened into the identifier lists of the constructors.
Finally, we will investigate how to assist users in the design process of their transformation rules; for instance by lifting \emph{schema matching} techniques~\cite{bernstein_model_2007, bellahsene2011schema} from relational to property graph schemas.

\onecolumn
\begin{multicols}{2}
\bibliographystyle{ACM-Reference-Format}
\bibliography{biblio}
\end{multicols}

\end{document}